\documentclass[11pt,letterpaper]{article}

\usepackage{geometry}
\usepackage{hyperref}
\usepackage{xspace}
\usepackage{array}
\usepackage{dcolumn}
\newcolumntype{d}[1]{D{.}{.}{#1}}
\usepackage{multirow}
\usepackage{stmaryrd}
\usepackage{tikz}
\usepackage{color}
\usetikzlibrary{plotmarks}

\usepackage{amsmath,amssymb,amsthm}
\newtheorem{theorem}{Theorem}
\newtheorem{lemma}{Lemma}

\usepackage[numbers]{natbib}
\usepackage{optparams,ifthen,xspace}

\makeatletter
\let\citeptemp\citep
\long\def\citep@[#1][#2]#3{%
	\ifthenelse{\equal{#1}{}}{%
		\ifthenelse{\equal{#2}{}}{%
			\citeptemp[][]{#3}}{%
			\citeptemp[][#2]{#3}}}{%
		\ifthenelse{\equal{#2}{}}{%
			(#1\citeptemp[][]{#3})}{%
			(#1\citeptemp[][]{#3}, #2)}}}
\renewcommand{\citep}{\optparams{\citep@}{[][]}}
\makeatother

\def\clap#1{\hbox to 0pt{\hss#1\hss}}
\def\mathllap{\mathpalette\mathllapinternal}  \def\mathclap{\mathpalette\mathclapinternal}
\def\mathllapinternal#1#2{\llap{$\mathsurround=0pt#1{#2}$}}

\def\mathclapinternal#1#2{\clap{$\mathsurround=0pt#1{#2}$}}

\newcount\Comments
\definecolor{purple}{rgb}{1,0,1}
\definecolor{orange}{rgb}{1,.5,0}
\newcommand{\kibitz}[2]{\ifnum\Comments=1\textcolor{#1}{#2}\fi}

\newcommand{\jl}[1]{\kibitz{blue}     {[John: #1]}}
\newcommand{\pd}[1]{\kibitz{green}     {[Paul: #1]}}

\newcommand{\ie}{i.e.,\xspace}
\newcommand{\eg}{e.g.,\xspace}

\newcommand{\expect}{\operatornamewithlimits{\mathbb{E}}}

\newcommand{\setprod}{\vartimes}
\newcommand{\reals}{\mathbb{R}}
\newcommand{\nats}{\mathbb{N}}

\newcommand{\bc}{\mathbf{c}}
\newcommand{\bd}{\mathbf{d}}

\newcommand{\bigO}{O}
\newcommand{\nfrac}[2]{\genfrac{}{}{0pt}{}{#1}{#2}}
\newcommand{\svmstruct}{\textit{SVM}$^{\textit{struct}}$\xspace}
\newcommand{\regret}{\mathit{rgt}}
\newcommand{\irv}{\mathit{irv}}

\newcommand{\Kpolyd}{K_{\mathit{polyd}}}
\newcommand{\Krbf}{K_{\mathit{RBF}}}

\newcommand{\secref}[1]{Section~\ref{#1}}
\newcommand{\figref}[1]{Figure~\ref{#1}}
\newcommand{\tabref}[1]{Table~\ref{#1}}

\newcommand{\lemref}[1]{Lemma~\ref{#1}}
\newcommand{\thmref}[1]{Theorem~\ref{#1}}

\newcommand{\appref}[1]{Appendix~\ref{#1}}

\begin{document}

\title{Payment Rules through Discriminant-Based Classifiers}

\author{%
	Paul D\"{u}tting\thanks{%
	\'Ecole Polytechnique F\'ed\'erale de Lausanne, Lausanne, Switzerland,
        Email: \texttt{paul.duetting@epfl.ch}
	}
	\and 
	Felix Fischer\thanks{%
	University of Cambridge, Cambridge, UK,
	Email: \texttt{fischerf@statslab.cam.ac.uk}
	} 
	\and Pichayut Jirapinyo\thanks{%
	Harvard University, Cambridge, MA, USA, 
	Email: \texttt{jirapinyo11@college.harvard.edu}
        }
	\and John K.~Lai\thanks{%
	Harvard University, Cambridge, MA, USA, 
	Email: \texttt{jklai@post.harvard.edu}
	}
	\and Benjamin Lubin\thanks{%
        Boston University, Boston, MA, USA,
	Email: \texttt{blubin@bu.edu}
	}
	\and David C.~Parkes\thanks{%
	Harvard University, Cambridge, MA, USA, 
	Email: \texttt{parkes@eecs.harvard.edu}
	}
}

\date{}

\maketitle

\begin{abstract}
In mechanism design it is typical to impose incentive compatibility
and then derive an optimal mechanism subject to this constraint. By
replacing the incentive compatibility requirement with the goal of
minimizing expected ex post regret, we are able to adapt statistical 
machine learning techniques to the design of payment rules.  This
computational approach to mechanism design is applicable to domains
with multi-dimensional types and situations where computational
efficiency is a concern.  Specifically, given an outcome rule and
access to a type distribution, we train a support vector machine with
a special discriminant function structure such that it implicitly
establishes a payment rule with desirable incentive properties.  We
discuss applications to a multi-minded combinatorial auction with a
greedy winner-determination algorithm and to an assignment problem
with egalitarian outcome rule.
Experimental results demonstrate both that the construction produces
payment rules with low ex post regret, and that penalizing
classification errors is effective in preventing failures of ex post
individual rationality.
\end{abstract}

\section{Introduction}

Mechanism design studies situations where a set of agents each hold
private information about their preferences over different outcomes.
The designer chooses a center that receives claims about such
preferences, selects and enforces an outcome, and optionally collects
payments. The classical approach is to impose \emph{incentive
  compatibility}, ensuring that agents truthfully report their
preferences in strategic equilibrium. Subject to this constraint, the goal
is to identify a mechanism, \ie a way of choosing an outcome and
payments based on agents' reports, that optimizes a given design
objective like social welfare, revenue, or some notion of fairness.

There are, however, significant challenges associated with this
classical approach. First of all, it can be analytically cumbersome to
derive optimal mechanisms for domains that are ``multi-dimensional''
in the sense that each agent's private information is described
through more than a single number, and few results are known in this
case.\footnote{One example of a multi-dimensional domain is a 
  combinatorial auction, where an agent's preferences are described
  by a numerical value for each of several different bundles of
  items.}
Second, incentive compatibility can be costly, in
that adopting it as a hard constraint can preclude mechanisms with
useful economic properties. For example, imposing the strongest form
of incentive compatibility, truthfulness in a dominant strategy
equilibrium or {\em strategyproofness}, necessarily leads to poor
revenue, vulnerability to collusion, and vulnerability to false-name
bidding in combinatorial auctions where valuations exhibit
complementarities among items~\citep{AuMi06a,rastegari11}.
A third difficulty occurs when the optimal mechanism has an outcome or
payment rule that is computationally intractable.

In the face of these difficulties, we adopt statistical machine
learning to automatically infer mechanisms with good incentive
properties. Rather than imposing incentive compatibility as a hard
constraint, we start from a given outcome rule and use machine
learning techniques to identify a payment rule that minimizes agents'
\emph{expected ex post regret} relative to this outcome rule. Here,
the ex post regret an agent has for truthful reporting in a given
instance is the amount by which its utility could be increased through
a misreport. While a mechanism with zero ex post regret for all inputs
is obviously strategyproof, we are not aware of any additional direct
implication in terms of equilibrium properties.\footnote{The expected
  ex post regret given a distribution over types provides an upper
  bound on the expected regret of an agent who knows its own type but
  has only distributional information on the types of other agents.
  The latter metric is also appealing, but does not seem to fit well
  with the generalization error of statistical machine learning. An
  emerging literature is developing various regret-based metrics for
  quantifying the incentive properties of
  mechanisms~\citep{PKE01a,DaMi08a,Lubi10a,carroll11}, and there also
  exists experimental support for a quantifiable measure of the {\em
    divergence} between the distribution on payoffs in a mechanism and
  that in a strategyproof reference mechanism like the VCG
  mechanism~\citep{LuPa09a}. An earlier literature had looked for
  approximate incentive compatibility or incentive compatibility in
  the large-market limit, see, \eg the recent survey by
  \citet{carroll11}.
Related to the general theme of relaxing incentive compatibility is
work of \citet{PaSo10a} that provides a qualitative ranking of
different mechanisms in terms of the number of manipulable instances, 
and work of \citet{Budi10a} that introduces an asymptotic, binary, 
design criterion regarding incentive properties in a large 
replica economy limit. Whereas the present work is constructive, 
the latter seek to explain which mechanisms are adopted in 
practice.}
Support for expected ex post regret as a quantifiable target for
mechanism design rather comes from a simple model of manipulation
where agents face a certain cost for strategic behavior. If this cost
is higher than the expected gain, agents can be assumed to behave
truthfully.
We do insist on mechanisms in which the price to an agent, conditioned
on an outcome, is independent of its report. This provides additional
robustness against manipulation in the sense that there is no local
price sensitivity.\footnote{\citet{ErKl10a} consider a metric that emphasizes
this property.}

Our approach is applicable to domains that are multi-dimensional or
for which the computational efficiency of outcome rules is a concern.
Given the implied relaxation of incentive compatibility, the intended
application is to domains in which incentive compatibility is
unavailable or undesirable for outcome rules that meet certain
economic and computational desiderata. The payment rule is learned on
the basis of a given outcome rule, and as such the framework is most
meaningful in domains where revenue considerations are secondary to
outcome considerations.

The essential insight is that the payment rule of a strategyproof
mechanism can be thought of as a classifier for predicting the
outcome: the payment rule implies a price to each agent for each
outcome, and the selected outcome must be one that simultaneously
maximizes reported value minus price for every agent. By limiting
classifiers to discriminant functions\footnote{A discriminant function
  can be thought of as a way to distinguish between different outcomes
  for the purpose of making a prediction.} with this
``value-minus-price'' structure, where the price can be an arbitrary
function of the outcome and the reports of other agents, we obtain a
remarkably direct connection between multi-class classification and
mechanism design.
For an appropriate loss function, the discriminant function of a
classifier that minimizes generalization error over a hypothesis class
has a corresponding payment rule that minimizes expected ex post
regret among all payment rules corresponding to classifiers in this
class.
Conveniently, an appropriate method exists for multi-class
classification with large outcome spaces that supports the specific
structure of the discriminant function, namely the method of {\em
  structural support vector machines}~\citep{TJHA05a,JFY09a}. Just like standard support vector machines, it allows us to
adopt non-linear kernels, thus enabling price functions that depend in
a non-linear way on the outcome and on the reported types of other
agents.

In illustrating the framework, we focus on two situations where
strategyproof payment rules are not available: a greedy outcome rule for a multi-minded combinatorial auction in which each agent is interested in a constant number of bundles, and an assignment problem with an egalitarian outcome rule, \ie an outcome rule that maximizes the
minimum value of any agent. The experimental results we obtain
are encouraging, in that they demonstrate low expected ex post regret
even when the $0/1$ classification accuracy is only moderately good, and
in particular better regret properties than those obtained through
simple VCG-based payment rules that we adopt as a baseline. In
addition, we give special consideration to the failure of ex post
individual rationality, and introduce methods to bias the
classifier to avoid these kinds of errors as well as post hoc
adjustments that eliminate them.
As far as scalability is concerned, we emphasize that the
computational cost associated with our approach occurs offline during
training. The learned payment rules have a succinct description and
can be evaluated quickly in a deployed mechanism.

\subsection*{Related Work}

\citet{Conitzer02} introduced the agenda of {\em automated mechanism
  design} (AMD), which formulates mechanism design as an optimization
problem.  The output is the {\em description} of a mechanism, \ie an
explicit mapping from types to outcomes and payments. AMD is
intractable in general, as the type space can be exponential in both
the number of agents and the number of items, but progress has
recently been made in finding approximate solutions for domains with
additive value structure and symmetry assumptions, and adopting
Bayes-Nash incentive compatibility (BIC) as the goal~\cite{CaDaWe11}.
Another approach is to search through a parameterized space of
incentive-compatible mechanisms~\citep{Guo10}.  %

A parallel literature allows outcome rules to be represented by {\em
  algorithms}, like our work, and thus extends to richer domains.
\citet{Lavi05} employ LP relaxation to obtain mechanisms satisfying
BIC for set-packing problems, achieving worst-case approximation
guarantees for combinatorial auctions. \citet{Hartline10a} and
\citet{Hartline10b} propose a general approach, applicable to both
single-parameter and multi-parameter domains, for converting any approximation
algorithm into a mechanism satisfying BIC that has essentially the same
approximation factor with respect to social welfare. This approach
differs from ours in that it adopts BIC as a target rather than the
minimization of expected ex post regret. In addition, it evaluates the
outcome rule on a number of randomly perturbed replicas of the
instance that is polynomial in the size of a discrete type space,
which is infeasible for combinatorial auctions where this size is
exponential in the number of items.  The computational requirements of
our trained rule are equivalent to that of the original outcome rule.

\citet{Laha09a,Laha10a} also adopts a kernel-based approach for
combinatorial auctions, but focuses not on learning a payment rule for
a given outcome rule but rather on solving the winner determination
and pricing problem for a given instance of a combinatorial
auction. \citeauthor{Laha09a} introduces the use of kernel methods to
compactly represent non-linear price functions, which is also present
in our work, but obtains incentive properties more indirectly through
a connection between regularization and price sensitivity.

\section{Preliminaries}
\label{sec:preliminaries}

A mechanism design problem is given by a set $N=\{1,2,\dots,n\}$ of
\emph{agents} that interact to select an element from a set
$\Omega\subseteq\setprod_{i\in N}\Omega_i$ of \emph{outcomes}, where
$\Omega_i$ denotes the set of possible outcomes for agent $i\in N$.
Agent $i\in N$ is
associated with a \emph{type} $\theta_i$ from a set $\Theta_i$ of
possible types, corresponding to the private information available to
this agent. We write $\theta=(\theta_1,\dots,\theta_n)$ for a profile
of types for the different agents, $\Theta=\setprod_{i\in N}\Theta_i$
for the set of possible type profiles, and $\theta_{-i}\in\Theta_{-i}$
for a profile of types for all agents but~$i$. Each agent $i\in N$ is
further assumed to employ preferences over $\Omega_i$, represented by
a \emph{valuation function} $v_i:\Theta_i\times\Omega_i
\rightarrow\reals$.
We assume that for all $i\in N$ and $\theta_i\in \Theta_i$ there
exists an outcome $o\in\Omega$ with $v_i(\theta_i,o_i) = 0$.

A \emph{(direct) mechanism} is a pair $(g,p)$ of an \emph{outcome
  rule} $g:\Theta\rightarrow\setprod_{i\in N}\Omega_i$ and a
\emph{payment rule} $p:\Theta\rightarrow\reals^n_{\ge 0}$. The
intuition is that the agents reveal to the mechanism a type profile
$\theta\in\Theta$, possibly different from their true types, and the
mechanism chooses outcome $g(\theta)$ and charges each agent~$i$ a
payment of $p_i(\theta)=(p(\theta))_i$.
We assume \emph{quasi-linear preferences}, so the \emph{utility} of
agent~$i$ with type $\theta_i\in\Theta_i$ given a profile
$\theta'\in\Theta$ of revealed types is
$u_i(\theta',\theta_i)=v_i(\theta_i, g_i(\theta'))-p_i(\theta')$,
where $g_i(\theta)=(g(\theta)_i)$ denotes the outcome for agent~$i$.
A crucial property of mechanism $(g,p)$ is that its outcome rule is
\emph{feasible}, \ie that $g(\theta)\in\Omega$ for all
$\theta\in\Theta$.

Outcome rule $g$ satisfies {\em consumer sovereignty} if for all $i\in
N$, $o_i\in\Omega_i$, and $\theta'_{-i}\in\Theta_{-i}$, there exists
$\theta'_i\in\Theta_i$ such that $g_i(\theta'_i,\theta'_{-i})=o_i$;
and {\em reachability of the null outcome} if for all $i\in N$,
$\theta_i\in\Theta_i$, and $\theta'_{-i}\in\Theta_{-i}$, there exists
$\theta'_i\in\Theta_i$ such that
$v_i(\theta_i,g_i(\theta'_{i},\theta'_{-i}))=0$.

Mechanism $(g,p)$ is \emph{dominant strategy
  incentive compatible}, or \emph{strategyproof}, if each agent
maximizes its utility by reporting its true type, irrespective of the
reports of the other agents, \ie if for all $i\in N$,
$\theta_i\in\Theta_i$, and $\theta'=
(\theta'_i,\theta'_{-i})\in\Theta$,
$u_i((\theta_i,\theta'_{-i}),\theta_{i}) \geq
u_i((\theta'_i,\theta'_{-i}),\theta_{i})$;
it satisfies \emph{individual rationality} (IR) if agents reporting 
their true types are guaranteed non-negative utility, \ie if for 
all $i\in N$, $\theta_i\in \Theta_i$, and $\theta'_{-i}\in\Theta_{-i}$,
$u_i((\theta_i,\theta'_{-i}),\theta_i) \geq 0$.
Observe that given reachability of the null outcome, strategyproofness 
implies individual rationality.

It is known that a mechanism $(g,p)$ is strategyproof if and only if
the payment of an agent is independent of its reported type and the
chosen outcome simultaneously maximizes the utility of all agents, \ie
if for every $\theta\in\Theta$, %
\begin{align}
  p_i(\theta) &= t_i(\theta_{-i},g_i(\theta)) && 
  \text{for all} \ i \in N, \ \text{and} \label{eq:1} \\ 
  g_i(\theta) &\in \arg\max_{\mathclap{o'_i \in \Omega_i}}
     \bigl(v_i(\theta_i, o'_i)-t_i(\theta_{-i},o_i')\bigr) && 
  \text{for all} \ i\in N, \label{eq:2} 
\end{align}
for a {\em price function}
$t_i:\Theta_{-i}\times\Omega_i\rightarrow\reals$.
This simple characterization is crucial for the main results in the
present paper, providing the basis with which the discriminant function of
a classifier can be used to induce a payment rule.

In addition, a direct characterization of strategyproofness in terms
of monotonicity properties of outcome rules explains which outcome
rules can be associated with a payment rule in order to be
``implementable'' within a strategyproof
mechanism~\citep{saks2005weak,ashlagi2010monotonicity}.
These monotonicity properties provide a fundamental constraint on when
our machine learning framework can hope to identify a payment rule
that provides full strategyproofness.  

We quantify the degree of strategyproofness of a mechanism in terms of
the {\em regret} experienced by an agent when revealing its true type,
\ie the potential gain in utility by revealing a different type
instead.
Formally, the \emph{ex post regret} of agent $i\in N$ in mechanism
$(g,p)$, given true type $\theta_i\in\Theta_i$ and reported types
$\theta'_{-i}\in\Theta_{-i}$ of the other agents, is
\[
\regret_i(\theta_i,\theta'_{-i}) = \max_{\mathclap{\theta'_i\in\Theta_i}} u_i\bigl((\theta'_i,\theta'_{-i}),\theta_i\bigr) - u_i\bigl((\theta_i,\theta'_{-i}),\theta_i\bigr) .
\]
Analogously, the {\em ex post violation of individual rationality} of
agent $i\in N$ in mechanism $(g,p)$, given true type $\theta_i\in\Theta_i$
and reported types $\theta'_{-i}\in\Theta_{-i}$ of the other agents,
is
\[
	\irv_i(\theta_i,\theta'_{-i}) = |\min(u_i((\theta_i,\theta'_{-i}), \theta_i),0)|.
\]

We consider situations where types
are drawn from a distribution with probability density function
$D:\Theta\rightarrow\reals$ such that $D(\theta)\geq 0$ and
$\int_{\theta\in\Theta}D(\theta)=1$.
Given such a distribution, and assuming that all agents report their
true types, the \emph{expected ex post regret} of agent $i\in N$ in
mechanism $(g,p)$ is
$\expect_{\theta\sim D} [\regret_i(\theta_i,\theta_{-i})]$.

Outcome rule $g$ is {\em agent symmetric} if for every permutation
$\pi$ of $N$ and all types $\theta,\theta'\in\Theta$ such that
$\theta_i=\theta'_{\pi(i)}$ for all $i\in N$,
$g_i(\theta)=g_{\pi(i)}(\theta')$ for all $i\in N$. Note that this
specifically requires that $\Theta_i=\Theta_j$ and $\Omega_i=\Omega_j$
for all $i,j\in N$.
Similarly, type distribution $D$ is {\em agent symmetric} if
$D(\theta)=D(\theta')$ for every permutation $\pi$ of $N$ and all
types $\theta,\theta'\in\Theta$ such that $\theta_i=\theta'_{\pi(i)}$
for all $i\in N$.
Given agent symmetry, a
price function $t_1: \Theta_{-1}\times\Omega_i\rightarrow\reals$ 
for agent~$1$ can be used to generate the payment rule $p$
for a mechanism $(g,p)$, with
\begin{align*}
	p(\theta)&=\bigl(t_1(\theta_{-1},g_1(\theta)), 
                         t_1(\theta_{-2},g_2(\theta)), \ldots,
                         t_1(\theta_{-n},g_n(\theta))\bigr),
\end{align*}
so that the expected ex post regret is the same for every agent.

We assume agent symmetry in the sequel, which precludes outcome rules
that break ties based on agent identity, but obviates the need to
train a separate classifier for each agent while also providing some
benefits in terms of presentation.  Because ties occur only 
with negligible probability in our experimental framework, the
experimental results are not affected by this assumption.

\section{Payment Rules from Multi-Class Classifiers}
\label{sec:defnclassification}

A \emph{multi-class classifier} is a function $h:X\rightarrow Y$,
where $X$ is an input domain and $Y$ is a discrete output domain. One
could imagine, for example, a multi-class classifier that labels a
given image as that of a dog, a cat, or some other animal. In the
context of mechanism design, we will be interested in classifiers that
take as input a type profile and output an outcome.  What
distinguishes this from an outcome rule is that we will impose
restrictions on the form the classifier can take.

Classification typically assumes an underlying target function
$h^\ast:X\rightarrow Y$, and the goal is to learn a classifier $h$
that minimizes disagreements with $h^\ast$ on a given input
distribution $D$ on $X$, based only on a finite set of {\em training
  data} $\{(x^1,y^1), \ldots, (x^\ell,y^\ell)\}=\{(x^1,h^\ast(x^1)),
\ldots, (x^\ell,h^\ast(x^\ell))\}$ with $x^1,\dots,x^\ell$ drawn from
$D$. This may be challenging because the amount of training data is
limited, or because $h$ is restricted to some hypothesis class
${\mathcal H}$ with a certain simple structure, \eg linear threshold
functions.  If $h(x)=h^\ast(x)$ for all $x\in X$, we say that $h$ is a
\emph{perfect classifier} for $h^\ast$.

We consider classifiers that are defined in terms of a {\em
  discriminant function} $f:X\times Y\rightarrow\reals$, such that
\[
	h(x)\in \arg\max_{y\in Y}f(x,y)
\]
for all $x\in X$. More specifically, we will be concerned with {\em
linear} discriminant functions of the form %
\[
	f_w(x,y) = w^T \psi(x,y) 
\]
for a weight vector $w\in\reals^m$ and a {\em feature map}
$\psi:X\times Y\rightarrow\reals^m$, where
$m\in\nats\cup\{\infty\}$.\footnote{We allow $w$ to have infinite
  dimension, but require the inner product between $w$ and $\psi(x,y)$
  to be defined in any case. Computationally the infinite-dimensional case
  is handled through the kernel trick, which is described in
  \secref{sec:kerneltrick}.}
The function $\psi$ maps input and output into an $m$-dimensional
space, which generally allows non-linear features to be expressed.

\subsection{Mechanism Design as Classification}
\label{sec:payment-rule}

Assume that we are given an outcome rule $g$ and access to a
distribution $D$ over type profiles, and want to design a
corresponding payment rule $p$ that gives the mechanism $(g,p)$ the
best possible incentive properties.
Assuming agent symmetry, we focus on a partial outcome rule $g_1:\Theta\rightarrow
\Omega_1$ and train a classifier to predict the outcome to agent~$1$.
To train a classifier, we generate examples by drawing a type profile
$\theta\in\Theta$ from distribution $D$ and applying outcome rule $g$ to
obtain the target class $g_1(\theta)\in \Omega_1$.

We impose a special structure on the hypothesis class. A classifier
$h_w:\Theta\rightarrow\Omega_1$ is {\em admissible} if it is defined
in terms of a discriminant function $f_w$ of the form
\[
	f_w(\theta,o_1) = w_1 v_1(\theta_1, o_1) + w^T_{-1} \psi(\theta_{-1},o_1)
\]
for weights $w$ such that $w_1\in\reals_{>0}$ and $w_{-1}\in\reals^m$,
and a feature map $\psi:\Theta_{-1}\times\Omega_1\rightarrow\reals^m$
for $m\in\mathbb{N}\cup\{\infty\}$.

The first term of $f_w(\theta,o_1)$ only depends on the type of
agent~$1$ and increases in its valuation for outcome $o_1$, while the
remaining terms ignore $\theta_1$ entirely. This restriction allows us
to directly infer agent-independent prices from a trained
classifier. For this, define the {\em associated price function} of an
admissible classifier $h_w$ as
\[
	t_{w}(\theta_{-1},o_1) = -\frac{1}{w_1} w^T_{-1} \psi(\theta_{-1},o_1),
\]
where we again focus on agent $1$ for concreteness. By agent symmetry, we
obtain the mechanism $(g,p_w)$ corresponding to classifier $h_w$ by
letting
\[
	p_w(\theta) = \bigl(t_w(\theta_{-1},g_1(\theta)), t_w(\theta_{-2},g_2(\theta)), \ldots,t_w(\theta_{-n},g_n(\theta))\bigr) .
\]

Even with admissibility, appropriate choices for the feature map
$\psi$ will produce rich families of classifiers, and thus ultimately
useful payment rules. Moreover, this form is compatible with
structural support vector machines, discussed in \secref{subsec:structural-svms}.

\subsection{Example: Single-Item Auction}
\label{subsec:single-item-example}

Before proceeding further, we illustrate the ideas developed so far in
the context of a single-item auction. In a single-item auction, the
type of each agent is a single number, corresponding to its value for
the item being auctioned, and there are two possible allocations
from the point of view of agent~$1$: one where it receives the item,
and one where it does not. Formally, $\Theta=\reals^n$ and
$\Omega_1=\{0,1\}$.

Consider a setting with three agents and a training set
\[(\theta^1,o_1^1) = ((1,3,5),0),\quad
	(\theta^2,o_1^2) = ((5,4,3),1),\quad
	(\theta^3,o_1^3) = ((2,3,4),0),
\]
and note that this training set is consistent with an {\em
  optimal} outcome rule, \ie one that assigns the item to an agent
with maximum value.
Our goal is to learn an admissible classifier
\[
	h_w(\theta) = \arg\max_{\mathllap{o_1\in\{0,1\}}}\; f_w(\theta, o_1) = \arg\max_{\mathclap{o_1\in\{0,1\}}}\; w_1 v_1(\theta_1, o_1) + w^T_{-1} \psi(\theta_{-1}, o_1)
\]
that performs well on the training set. Since there are only
two possible outcomes, the outcome chosen by $h_w$ is simply the one
with the larger discriminant. A classifier that is perfect on the
training data must therefore satisfy the following constraints:
\begin{align*}
	\allowdisplaybreaks
w_1 \cdot 0 + w^T_{-1} \psi((3,5),0) > w_1 \cdot 1 + w^T_{-1} \psi((3,5),1) , \\
w_1 \cdot 5 + w^T_{-1} \psi((4,3),1) > w_1 \cdot 0 + w^T_{-1} \psi((4,3),0) , \\
w_1 \cdot 0 + w^T_{-1} \psi((3,4),0) > w_1 \cdot 2 + w^T_{-1} \psi((3,4),1) .
\end{align*}
This can for example be achieved by setting $w_1 = 1$ and
\begin{equation}
\label{eq:2ndprice}
w_{-1}^T \psi((\theta_2, \theta_3), o_1) = 
\begin{cases}
		-\max(\theta_2, \theta_3) &\text{if $o_1=1$ and} \\
		0 & \text{if $o_1 = 0$.} 
\end{cases}
\end{equation}

Recalling our definition of the price function as $t_w(\theta_{-1}, o_1) =
-(1/w_1) w_{-1}^T \psi(\theta_{-1}, o_1)$, we see that this choice of $w$
and $\psi$ corresponds to the second-price payment rule. We will see
in the next section that this relationship is not a
coincidence.\footnote{In practice, we are limited in the machine
  learning framework to hypotheses that are linear in
  $\psi((\theta_2,\theta_3),o_1)$, and will not be able to guarantee 
that~\eqref{eq:2ndprice} holds exactly. In
  \secref{subsubsec:the-psi-function} we will see, however, that
  certain choices of $\psi$ allow for very complex hypotheses that can
  closely approximate arbitrary functions.}

\subsection{Perfect Classifiers and Implementable Outcome Rules}
\label{sec:regret}

We now formally establish a connection between implementable outcome
rules and perfect classifiers.
\begin{theorem}
	Let $(g,p)$ be a strategyproof mechanism with an agent
        symmetric outcome rule $g$, and let $t_1$ be the corresponding
        price function. Then, a perfect admissible classifier $h_w$
        for partial outcome rule $g_1$ exists if
        $\arg\max_{o_1\in\Omega_1}\left(v_1(\theta_1, o_1) -
          t_1(\theta_{-1},o_1))\right)$ is unique.
\end{theorem}
\begin{proof}
	By the first characterization of strategyproof mechanisms, $g$
	must select an outcome that maximizes the utility of agent~$1$
	at the current prices, \ie
\[
	g_1(\theta) \in \arg\max_{\mathclap{o_1\in\Omega_1}} ( v_1(\theta_i, o_1) -
  t_1(\theta_{-1},o_1)) .
\]
Consider the admissible discriminant
$f_{(1,1)}(\theta,o_1)=v_1(\theta_1, o_1)-t_1(\theta_{-1},o_1)$, which
uses the price function $t_1$ as its feature map. Clearly, the
corresponding classifier $h_{(1,1)}$ maximizes the same quantity as
$g_1$, and the two must agree if there is a unique maximizer.
\end{proof}

The relationship also works in the opposite direction: a perfect,
admissible classifier $h_w$ for outcome rule $g$ can be used to
construct a payment rule that turns $g$ into a strategyproof
mechanism.
\begin{theorem} \label{thm:exact-implies-sp}
	Let $g$ be an agent symmetric outcome rule,
        $h_w:\Theta\rightarrow\Omega_1$ an admissible classifier, and
        $p_w$ the payment rule corresponding to $h_w$. If $h_w$ is a
        perfect classifier for the partial outcome rule $g_1$, then
        the mechanism $(g,p_w)$ is strategyproof.
\end{theorem}

We prove this result by expressing the regret of an agent in mechanism
$(g,p_w)$ in terms of the discriminant function $f_w$.
Let $\Omega_i(\theta_{-i})\subseteq\Omega_i$ denote the set of partial
outcomes for agent $i$ that can be obtained under~$g$ given reported
types $\theta_{-i}$ from all agents but~$i$, keeping the dependence on
$g$ silent for notational simplicity.
\begin{lemma} \label{lem:regret}
	Suppose that agent~$1$ has type $\theta_1$ and that the other
	agents report types $\theta_{-1}$. Then the regret of
	agent~$1$ for bidding truthfully in mechanism $(g,p_{w})$ is
\[
	\frac{1}{w_1}\bigl(\max_{o_1\in\Omega(\theta_{-1})} f_{w}(\theta, o_1) -
f_{w}(\theta, g_1(\theta))\bigr).
\]
\end{lemma}
\begin{proof}
	We have
\begin{align*}
	\regret_{1}(\theta) 
	&= \hspace{0.35cm} \max_{\theta'_1 \in \Theta_1} \hspace{0.35cm} \bigl(v_1(\theta_1,
  g_1(\theta'_1,\theta_{-1})) - p_{w,1}(\theta'_1,\theta_{-1})\bigr) -
  \bigl(v_1(\theta_1, g_1(\theta))-p_{w,1}(\theta)\bigr) \\
	&= \max_{o_1\in \Omega_1(\theta_{-1})} \bigl(v_1(\theta_1,
  o_1)-t_{w}(\theta_{-1},o_1)\bigr) - \bigl(v_1(\theta_1, g_1(\theta))-t_{w}(\theta_{-1},g_1(\theta))\bigr) \\
  	&= \max_{o_1\in \Omega_1(\theta_{-1})} \bigl(v_1(\theta_1,
  o_1)+\frac{1}{w_1}w_{-1}^T \psi(\theta_{-1},o_1)\bigr) -
	\bigl(v_1(\theta_1, g_1(\theta))+\frac{1}{w_1}w_{-1}^T \psi(\theta_{-1},g_1(\theta))\bigr) \\
	&= \hspace{0.4cm} \frac{1}{w_1} \bigl(\max_{o_1\in \Omega_1(\theta_{-1})} f_{w}(\theta, o_1) - f_{w}(\theta,g_1(\theta))\bigr). \tag*{\raisebox{-1ex}{\qedhere}}
\end{align*} 
\end{proof}

\begin{proof}[Proof of \thmref{thm:exact-implies-sp}]
If $h_w$ is a perfect classifier, then the discriminant function $f_w$
satisfies $\arg\max_{o_1 \in \Omega_1}f_w(\theta,o_1) = g_1(\theta)$
for every $\theta\in\Theta$. Since
$g_1(\theta)\in\Omega_1(\theta_{-1})$, we thus have that
$\max_{o_1\in\Omega_1(\theta_{-1})} f_w(\theta, o_1) =
f_w(\theta,g_1(\theta))$. By \lemref{lem:regret}, the regret of
agent~$1$ for bidding truthfully in mechanism $(g,p_{w})$ is always
zero, which means that the mechanism is strategyproof.
\end{proof}

It bears emphasis that classifier $h_w$ is only used to derive the
payment rule $p_w$, while the outcome is still selected according to
$g$. In principle, classifier $h_w$ could be used to obtain an agent
symmetric outcome rule $g_w$ and, since $h_w$ is a perfect classifier
for itself, a strategyproof mechanism $(g_w,p_w)$. Unfortunately,
outcome rule $g_w$ is not in general feasible. Mechanism $(g,p_w)$, on
the other hand, is not strategyproof when $h_w$ fails to be a perfect 
classifier for $g$. While payment rule $p_w$ always satisfies the
agent-independence property~\eqref{eq:1} required
for strategyproofness, the ``optimization'' property~\eqref{eq:2} might be violated when $h_w(\theta)\neq
g_1(\theta)$.

\subsection{Approximate Classification and Approximate Strategyproofness}
\label{sec:minimize-generalization}

A perfect admissible classifier for outcome rule $g$ leads to a
payment rule that turns $g$ into a strategyproof mechanism. We now
show that this result extends gracefully to situations where no such
payment rule is available, by relating the \emph{expected} ex post
regret of a mechanism $(g,p)$ to a measure of the generalization error
of a classifier for~$g$.

Fix a feature map $\psi$, and denote by $\mathcal{H}_\psi$ the space
of all admissible classifiers with this feature map. The
\emph{discriminant loss} of a classifier $h_w\in\mathcal{H}_\psi$ with
respect to a type profile $\theta$ and an outcome $o_1\in\Omega_1$ is
given by
\[
	\Delta_w(o_1,\theta) = \frac{1}{w_1} \bigl(f_w(\theta,h_w(\theta)) - f_w(\theta,o_1)\bigr) .
\]
Intuitively the discriminant loss measures how far, in terms of the normalized
discriminant, $h_w$ is from predicting the correct outcome for type
profile~$\theta$, assuming the correct outcome is $o_1$. Note that
$\Delta(o_1,\theta)\geq 0$ for all $o_1\in\Omega_1$ and
$\theta\in\Theta$, and $\Delta(o_1,\theta)=0$ if
$o_1=h_w(\theta)$. Note further that $h_w(\theta)=h_{w'}(\theta)$ does
not imply that $\Delta_w(o_1,\theta)=\Delta_{w'}(o_1,\theta)$ for all
$o_1\in\Omega_1$: even if two classifiers predict the same outcome,
one of them may still be closer to predicting the correct outcome
$o_1$.

The \emph{generalization error} of classifier $h_w\in\mathcal{H}_\psi$
with respect to a type distribution $D$ and a partial outcome rule
$g_1:\Theta\rightarrow\Omega_1$ is then given by
\[
	R_{w}(D,g) = \int_{\theta\in\Theta}\Delta_w\bigl(g_1(\theta),\theta\bigr) D(\theta)d\theta.
\]
The following result establishes a connection between the
generalization error and the expected ex post regret of the
corresponding mechanism.
\begin{theorem}
Consider an outcome rule $g$, a space ${\mathcal H}_\psi$ of admissible classifiers, and a type distribution $D$.  Let $h_{w^*}
\in\mathcal{H}_\psi$ be a classifier that minimizes generalization
error with respect to $D$ and $g$ among all classifiers in
$\mathcal{H}_\psi$. Then the following holds:
\begin{enumerate}
	\item If $g$ satisfies consumer sovereignty, then
	$(g,p_{w^*})$ minimizes expected ex post regret with respect
	to $D$ among all mechanisms $(g, p_w)$ corresponding to
	classifiers \mbox{$h_w\in\mathcal{H}_\psi$}.  \label{itm:errgt1}
	\item Otherwise, $(g,p_{w^*})$ minimizes an upper bound on
          expected ex post regret with respect to $D$ amongst all
          mechanisms $(g,p_w)$ corresponding to classifiers
          $h_w\in\mathcal{H}_\psi$.  \label{itm:errgt2}
\end{enumerate}
\end{theorem}
\begin{proof}
	For the second property, observe that
	\begin{align*}
	\Delta_w(g_1(\theta),\theta)&=\frac{1}{w_1}\bigl(f_w(\theta,h_w(\theta))-f_w(\theta,g_1(\theta))\bigr) \\
	& = \frac{1}{w_1}\bigl(\max_{o_1\in\Omega_1} f_w(\theta,o_1) -
	f_w(\theta,g_1(\theta))\bigr) \\ 
	&\geq \frac{1}{w_1}\bigl(\max_{o_1 \in \Omega(\theta_{-1})} f_{w}(\theta,o_1) -
f_{w}(\theta,g_1(\theta))\bigr) 
	= \regret_1(\theta), 
	\end{align*}
	where the last equality holds by \lemref{lem:regret}. If $g$
        satisfies consumer sovereignty, then the inequality holds with
        equality, and the first property follows as well.
\end{proof}
Minimization of expected regret itself, rather than an upper bound, can also be achieved if the learner has access to the set $\Omega_1(\theta_{-1})$ for every $\theta_{-1}\in\Theta_{-1}$.

\section{A Solution using Structural Support Vector Machines}
\label{sec:sol-via-svm}

In this section we discuss the method of {\em structural support vector 
machines} (structural SVMs)~\citep{TJHA05a,JFY09a}, and show how it can be 
adapted for the purpose of learning classifiers with admissible discriminant 
functions.

\subsection{Structural SVMs}
\label{subsec:structural-svms}

Given an input space $X$, a discrete output space $Y$, a target
function $h^\ast:X \rightarrow Y$, and a set of \emph{training
  examples} $\{ (x^1, h^\ast(x^1)), \ldots, (x^\ell, h^\ast(x^\ell))
\} = \{ (x^1, y^1), \ldots, (x^\ell, y^\ell)\}$, structural SVMs learn
a multi-class classifier $h$ that on input $x\in X$ selects an output $y\in Y$
that maximizes $f_w(x, y) = w^T \psi(x, y)$.
For a given feature map $\psi$, the training problem is to find a vector
$w$ for which $h_w$ has low generalization error.

Given examples $\{(x^1, y^1), \ldots, (x^\ell, y^\ell)\}$,
training is achieved by solving the following convex optimization
problem:
\begin{align*}
	\min_{\mathclap{w,\xi\geq 0}}\;\;& \frac{1}{2}w^Tw +
	\frac{C}{\ell}\sum_{k=1}^\ell \xi^k \tag{Training Problem 1} \\  
	\text{s.t.}\;\; 
	& w^T\big(\psi(x^k,y^k)-\psi(x^k,y)\big) \geq {\mathcal L}(y^k,y)-\xi^k 
	\quad\text{for all $k=1,\ldots,\ell$, $y\in Y$} \\
	&\xi^k \geq 0 \quad\text{for all $k=1,\ldots,\ell$.}
\end{align*}
The goal is to find a weight vector $w$ and slack variables
$\xi^k$ such that the objective function is minimized while satisfying the
constraints. The learned weight vector $w$ parameterizes the
discriminant function $f_{w}$, which in turn defines the 
classifier $h_{w}$. The $k$th constraint states that the
value of the discriminant function on $(x^k, y^k)$ should exceed the
value of the discriminant function on $(x^k, y)$ by at least
${\mathcal L}(y^k, y)$, where ${\mathcal L}$ is a loss function that
penalizes misclassification, with ${\mathcal L}(y,y)=0$ and 
${\mathcal L}(y,y')\geq 0$ for all $y,y'\in Y$.
We generally use a $0/1$ loss function, but consider an alternative in
\secref{sec:IRLoss} to improve ex post IR properties.
Positive values for the slack variables $\xi^k$ allow the weight
vector to violate some of the constraints.

The other term in the objective, the squared norm of $w$, penalizes
scaling of $w$. This is necessary because scaling of $w$ can arbitrarily 
increase the margin
between $f_w(x^k, y^k)$ and $f_w(x^k, y)$ and make the constraints
easier to satisfy. Smaller values of $w$, on the other hand, increases 
the ability of the learned classifier to generalize by decreasing the 
propensity to over-fit to the training data.
Parameter $C$ is therefore a regularization parameter: larger values of $C$ 
encourage small $\xi^k$ and larger $w$, such that more points 
are classified correctly, but with a smaller margin.

\subsubsection{The Feature Map and the Kernel Trick}
\label{sec:kerneltrick}
\label{subsubsec:the-psi-function}

Given a feature map $\psi$, the \emph{feature vector} $\psi(x,y)$ for 
$x\in X$ and $y\in Y$ provides an alternate representation of the 
input-output pair $(x,y)$. It is useful to consider feature maps $\psi$ for 
which $\psi(x,y)=\phi(\chi(x,y))$, where $\chi:X\times Y\rightarrow\reals^s$ 
for some $s\in{\mathbb N}$ is an {\em attribute map} that combines $x$ and $y$ 
into a single {\em attribute vector} $\chi(x,y)$ compactly representing the 
pair, and $\phi:\reals^s\rightarrow\reals^m$ for $m>s$ maps the attribute 
vector to a higher-dimensional space in a non-linear way. In this way, SVMs 
can achieve non-linear classification in the original space.

While we work hard to keep $s$ small, the so-called {\em kernel trick} means 
that we do not have the same problem with $m$: it turns out that in the dual
of Training Problem~1, $\psi(x,y)$ only appears in an inner product of the 
form $\langle\psi(x,y), \psi (x', y')\rangle$, or, for a
decomposable feature map, $\langle \phi(z), \phi(z') 
\rangle$ where $z=\chi(x,y)$ and $z'=\chi(x',y')$. For computational 
tractability it therefore 
suffices that this inner product can be computed efficiently, and the 
``trick'' is to choose $\phi$ such that 
$\langle\phi(z),\phi(z')\rangle=K(z,z')$ for a simple closed-form function 
$K$, known as the {\em kernel}.

In this paper, we consider {\em polynomial kernels} $\Kpolyd$,
parameterized by $d\in\mathbb{N}^+$, and {\em radial basis function (RBF) 
kernels} $\Krbf$, parameterized by $\gamma=1/(2\sigma^2)$ for $\sigma\in \reals^+$:
\begin{align*}
&\Kpolyd(z,z') = (z \cdot z')^d, \\
&\Krbf(z,z') = \exp\left( - \gamma \left( \lVert z \rVert^2 + \lVert z' \rVert^2 - 2 z \cdot z' \right) \right).
\end{align*}
Both polynomial and RBF kernels use the standard inner product of their
arguments, so their efficient computation requires that $\chi(x,y) \cdot\chi(x,y')$ can be computed efficiently.

\subsubsection{Dealing with an Exponentially Large Output Space}

Training Problem~1 has $\Omega(|Y|\ell)$ constraints, where $Y$ is the output
space and $\ell$ the number of training instances, and enumerating all of them
is computationally prohibitive when $Y$ is large. \citet{JFY09a} address this
issue for structural SVMs through constraint generation: starting from an empty
set of constraints, this technique iteratively adds a constraint that is
maximally violated by the current solution until that violation is below a
desired threshold $\epsilon$. \citeauthor{JFY09a} show that this will happen after no
more than $\bigO(\frac{C}{\epsilon})$ iterations, each of which requires
$\bigO(\ell)$ time and memory. However, this approach assumes the
existence of an efficient separation oracle, which given a weight vector $w$ and
an input $x$ finds an output $y \in \arg\max_{y′\in Y} f_w(x,y′)$. The existence
of such an oracle remains an open question in application to combinatorial
auctions; see \secref{subsubsec:exp-large-y} for additional discussion.

\subsubsection{Required Information}
\label{subsec:required-info}

In summary, the use of structural SVMs requires specification of the following:
\begin{enumerate}
  \item The input space $X$, the discrete output space $Y$, and
  examples of input-output pairs.
  \item An attribute map $\chi: X \times Y \rightarrow \mathbb{R}^s$.
  This function generates an attribute vector that combines the input
  and output data into a single object.  \label{required-feature-map}
  \item A kernel function $K(z,z')$, typically chosen from a
  well-known set of candidates, \eg polynomial or RBF.  The kernel
  implicitly calculates the inner product $\langle \phi(z), \phi(z')
  \rangle$, \eg between a mapping of the inputs into a high
  dimensional space.
  \item If the space $Y$ is prohibitively large, a routine that allows
  for efficient separation, \ie a function that computes
  $\arg\max_{y \in Y} f_w(x, y)$ for a given $w, x$.
\end{enumerate}
In addition, the user needs to stipulate particular
training parameters, such as the regularization parameter $C$, and the
kernel parameter $\gamma$ if the RBF kernel is being used.

\subsection{Structural SVMs for Mechanism Design}
\label{subsec:svms-for-md}

We now specialize structural SVMs such that their learned discriminant
function will manifest as a payment rule for a given symmetric
outcome function $g$ and distribution $D$.  In this application, the
input domain $X$ is the space of type profiles $\Theta$, and the
output domain $Y$ is the space $\Omega_1$ of outcomes for agent~$1$.  Thus
we construct training data by sampling $\theta \sim D$ and applying
$g$ to these inputs: $\{(\theta^1, g_1(\theta^1)), \ldots,
(\theta^\ell, g_1(\theta^\ell))\} = \{ (\theta^1, o_1^1), \ldots,
(\theta^\ell, o_1^\ell) \}$.
For admissibility of the learned hypothesis $h_w(\theta) = \arg\max_{o_1 \in
  \Omega_1} w^T\psi(\theta,o_1)$, we require that
\begin{align*} %
\psi(\theta,o_1) = (v_1(\theta_1, o_1), \psi'(\theta_{-1}, o_1))
\end{align*}
When learning payment rules, we therefore use an attribute map
$\chi':\Theta_{-1} \times \Omega_1 \rightarrow \reals^s$ rather than
$\chi:\Theta \times \Omega_1 \rightarrow \reals^s$, and the
kernel $\phi'$ we specify will only be applied to the output of
$\chi'$. This results in the following more specialized training problem:
\begin{align*}
	\min_{\mathclap{w,\xi\geq 0}}\;\;& \frac{1}{2}w^Tw +
	\frac{C}{\ell}\sum_{k=1}^\ell \xi^k \tag{Training Problem 2} \\
	\text{s.t.}\;\; 
	& (w_1 v_1(\theta_1^k,o_1^k) + w_{-1}^T \psi'(\theta_{-1}^k, o_1^k)) - (w_1 v_1(\theta_1^k,o_1) + w_{-1}^T \psi'(\theta_{-1}^k, o_1)) \geq
    {\mathcal L}(o_1^k, o_1) -\xi^k \\
	& \phantom{(w_1 v_1(\theta_1^k,o_1^k) + w_{-1}^T \psi'(\theta_{-1}^k, o_1^k)) - (w_1 v_1(\theta_1^k,o_1) +}\;\;
	\text{for all $k=1,\dots,\ell$, $o_1\in\Omega_1$} \\
	&\xi^k \geq 0 \quad\text{for all $k=1,\dots,\ell$.}
\end{align*}

If $w_1 > 0$ then the weights $w$ together with the feature map
$\psi'$ define a price function $t_w(\theta_{-1},o_1) = - (1/w_1)
w_{-1}^T \psi'(\theta_{-1},o_1)$ that can be used to define payments
$p_w(\theta)$, as described in \secref{sec:payment-rule}.
In this case, we can also relate the regret in the induced mechanism
$(g,p_w)$ to the classification error as described in
\secref{sec:regret}.

\begin{theorem}
\label{thm:slack-to-regret}
Consider training data $\{ (\theta^1, o_1^1), \ldots, (\theta^\ell,
o_1^\ell) \}$.  Let $g$ be an outcome function such that
$g_1(\theta^k) = o_1^k$ for all $k$.  Let $w, \xi^k$ be the weight
vector and slack variables output by Training Problem 2, with $w_1 >
0$.  Consider corresponding mechanism $(g, p_w)$.  For each
$\theta^k$,
\begin{equation*}
		\regret_1(\theta^k) \leq \frac{1}{w_1} \xi^k
	\end{equation*}
\end{theorem}
\begin{proof} 
Consider input $\theta^k$. The constraints in the training problem
impose that for every outcome $o_1 \in \Omega_1$,
\begin{equation*}
	w_1 v_1(\theta_1^k,o_1^k) + w_{-1}^T \psi'(\theta_{-1}^k,o_1^k) -
    \big(w_1 v_1(\theta_1^k,o_1) + w_{-1}^T \psi'(\theta_{-1}^k, o_1)\big) \geq
    {\mathcal L}(o_1^k, o_1) - \xi^k
\end{equation*}
Rearranging,
\begin{align*}
	  \xi^k &\geq {\mathcal L}(o_1^k, o_1)  +
    \big(w_1 v_1(\theta_1^k,o_1) + w_{-1}^T \psi'(\theta_{-1}^k, o_1)\big) -
    \big(w_1 v_1(\theta_1^k,o_1^k) + w_{-1}^T \psi'(\theta_{-1}^k, o_1^k)\big)\\
\Rightarrow \xi^k& \geq {\mathcal L}(o_1^k, o_1) + f_w(\theta^k, o_1) - f_w(\theta^k, o_1^k)\\
	\intertext{This inequality holds for every $o_1 \in \Omega_1$, so}
 \xi^k& \geq \max_{o_1 \in \Omega_1} \left( {\mathcal L}(o_1^k, o_1) + f_w(\theta^k, o_1) -
 f_w(\theta^k, o_1^k) \right)\\
 & \geq \max_{o_1 \in \Omega_1} \left( f_w(\theta^k, o_1) -
 f_w(\theta^k, o_1^k) \right)\\
	&\geq w_1 \regret_1(\theta^k)
\end{align*}
where the second inequality holds because ${\mathcal L}(o_1^k,
o_1)\geq 0$, and the final inequality follows from
\lemref{lem:regret}. This completes the proof.
\end{proof}

\pd{David suggested to say something about consistency here}

We choose not to enforce $w_1>0$ explicitly in Training Problem~2,
because adding this constraint leads to a dual problem that references
$\psi'$ outside of an inner product and thus makes computation of all
but linear or low-dimensional polynomial kernels prohibitively
expensive.
Instead, in our experiments we simply discard hypotheses where the result of
training is $w_1 \leq 0$.  This is sensible since the discriminant function
value should increase as an agent's value increases, and negative values of
$w_1$ typically mean that the training parameter $C$ or the kernel parameter
$\gamma$ (if the RBF kernel is used) are poorly chosen.  It turns out that $w_1$
is indeed positive most of the time, and for every experiment a majority of the
choices of $C$ and $\gamma$ yield positive $w_1$ values.  For this reason, we do
not expect the requirement that $w_1>0$ to be a problem in
practice.\footnote{For multi-minded combinatorial auctions, $1049/1080> 97\%$ of
the trials had positive $w_1$, for the assignment problem all of the trials did;
see \secref{sec:applying} for details.}

\subsubsection{Payment Normalization}

One issue with the framework as stated is that the payments $p_w$
computed from the solution to Training Problem~2 could be negative.

We solved this problem by normalizing payments, using a \emph{baseline 
outcome} $o_b$: if there exists an outcome $o'$ such that 
$v_1(\theta_1,o')=0$ for every $\theta_1$, this ``null outcome'' is used as 
the baseline; otherwise, we use the outcome with the lowest payment.  Let 
$t_w(\theta_{-1}, o_1)$ be the price function corresponding to the solution 
$w$ to Training Problem~2. Adopting the baseline outcome, the {\em normalized 
payments} 
$t_w'(\theta_{-1}, o_1)$ are defined
as
\[
t_w'(\theta_{-1}, o_1) = \max(0, t_w(\theta_{-1}, o_1) - t_{w}(\theta_{-1},
o_b)) .
\]
Note that $o_b$ is only a function of $\theta_{-1}$, even when there is no 
null outcome, so $t_w'$ is still only a function of
$\theta_{-1}$ and $o_1$.

\subsubsection{Individual Rationality Violation}
\label{subsubsec:ir-fixes}

Even after normalization, the learned payment rule $p_w$ may not satisfy
IR. We offer three solutions to this problem, which can be used in
combination.

\paragraph{Payment offsets}
One way to decrease the rate of IR violation is to add a payment
offset, which decreases all payments (for all type reports) by a given
amount.  We apply this payment offset to all bundles other than $o_b$; as 
with payment normalization, the adjusted payment is set to $0$ if it is 
negative.\footnote{It is again crucial that $o_b$
  depends only on $\theta_{-1}$, so that the payment remains
  independent of $\theta_1$ given~$o_1$.}
Note that payment offsets decrease IR violation, but may increase regret.
For instance, suppose there are only two outcomes $o_{11}, o_{12}$, where $o_{12}$
is the null outcome. Suppose agent 1 values $o_{11}$ at 5 and receives
the null outcome if he reports truthfully. Suppose further that
payments $t_w$ are 7 for $o_{11}$ and 0 for the null outcome.  With no
payment offset, the agent experiences no regret, since he receives
utility 0 from the null outcome, but negative utility from $o_{11}$.
However, if the payment offset is greater than 2, the agent's regret
becomes positive (assuming consumer sovereignty) because he could have
reported differently and received $o_{11}$ and received positive
utility.

\paragraph{Adjusting the loss function ${\mathcal L}$}
\label{sec:IRLoss}
We incur an IR violation when there is a null outcome $o_{\mathit{null}}$ 
such that $g_1(\theta) \ne o_{\mathit{null}}$ and 
$f_w(\theta,o_{\mathit{null}}) > f_w(\theta,g_1(\theta))$ for some type 
$\theta$, assuming truthful reports.
This happens because $f_w(\theta, o_1)$ is a scaled version of the
agent's utility for outcome $o_1$ under payments $p_w$.  If the
utility for the null outcome is greater than the utility for
$g_1(\theta)$, then the payment $t_w(\theta_{-1}, g_1(\theta))$ must
be greater than $v_1(\theta_1, g_1(\theta))$, causing an IR violation. 
We can discourage these types of errors by modifying the
constraints of Training Problem~2: when $o_1^k \ne o_{\mathit{null}}$ 
and $o_1=o_{\mathit{null}}$, we can increase ${\mathcal L}(o_1^k, o_1)$ to 
heavily penalize misclassifications of this type.  With a larger ${\mathcal
  L}(o_1^k, o_1)$, a larger $\xi^k$ will be required if $f_w(\theta,
o_1^k) < f_w(\theta,o_{\mathit{null}})$.  As with payment offsets, this
technique will decrease IR violations but is not guaranteed to
eliminate all of them.  In our experimental results, we refer to this 
as the {\em null loss fix}, and the null loss refers to the value
we choose for ${\mathcal L}(o_1^k,o_{\mathit{null}})$ where $o_1^k \ne
o_{\mathit{null}}$.

\paragraph{Deallocation}
In settings that have a null outcome and are {\em downward closed}
(\ie settings where a feasible outcome $o$ remains feasible if $o_i$ is 
replaced with the null outcome), we modify the function $g$ to allocate the 
null outcome whenever the price function $t_w$ creates an IR violation. This 
reduces ex post regret and in particular ensures ex post IR. On the other 
hand, the total value to the agents necessarily decreases under the
modified allocation. In our experimental results, we refer to this as the 
{\em deallocation fix}.

\section{Applying the Framework}
\label{sec:applying}

In this section, we discuss the application of our framework to two domains: 
multi-minded combinatorial auctions and egalitarian welfare in the assignment problem.

\subsection{Multi-Minded Combinatorial Auctions}
\label{subsec:application-to-cas}

A combinatorial auction allocates items $\{1,\ldots,r\}$
among $n$ agents, such that each agent receives a possibly empty subset
of the items. The outcome space $\Omega_i$ for agent $i$ thus is the set of all subsets of the $r$ items, and the type of agent $i$ can be represented by a vector $\theta_i\in\Theta_i=\mathbb{R}^{2^r}$ that specifies its value for each possible bundle. The set of possible type profiles is then $\Theta=\mathbb{R}^{2^r n}$, and the value $v_i(\theta_i,o_i)$ of agent $i$ for bundle $o_i$ is equal to the entry in $\theta_i$ corresponding to $o_i$. 
We require that valuations are monotone, such that $v_i(\theta_i, o_i) \geq v_i(\theta_i, o_i')$ for all $o_i,o_i'\in\Omega_i$ with $o_i'\subseteq o_i$, and normalized such that $v_i(\theta_i,\emptyset)=0$.
Assuming agent symmetry and adopting the view of agent~$1$, the partial outcome rule $g_1:\Theta\rightarrow\Omega_1$ specifies the bundle 
$g_1(\theta)$ allocated to agent $1$; we require feasibility, so 
that no item is allocated more than once.

In a multi-minded CA, each agent is interested in at most $b$ bundles
for some constant $b$. The special case where $b=1$ is called a
single-minded CA.
In our framework, the restriction to multi-minded CAs leads to a number of 
computational advantages. First, valuation profiles and thus the training data 
can be represented in a compact way, by explicitly writing down the valuations 
for the constant number of bundles each agent is interested in. Second, inner 
products between valuation profiles, which are required to apply the kernel 
trick, can be computed in constant time.

\subsubsection{Attribute Maps}
\label{subsubsec:attributemap}

To apply structural SVMs to multi-minded CAs, we need to specify an 
appropriate attribute map $\chi$. In our experiments we use two attribute 
maps $\chi_1:\Theta_{-1}\times\Omega_1\rightarrow\reals^{2^r(2^r(n-1))}$ and 
$\chi_2:\Theta_{-1}\times\Omega_1\rightarrow\reals^{2^{r}(n-1)}$, which are 
defined as follows:
\begin{equation*}
\chi_1(\theta_{-1}, o_1) = \left[
	\begin{array}{@{}c@{}}
		0 \\
		\cdots \\
	  0 \\
		\theta_{-1} \\
	  0 \\
		\cdots \\
	  0
	\end{array} \right] \hspace*{-1.5em}
	\begin{array}{l}
		\left.\begin{array}{c}
			\\ \phantom{\cdots} \\ \\
		\end{array} \right\} dec(o_1)(2^r(n-1)) \\
		\begin{array}{c}
			\\ 
		\end{array} \\
		\left.\begin{array}{c}
			\\ \phantom{\cdots} \\ \\
		\end{array} \right\} (2^r- dec(o_1)-1)(2^r(n-1))
	\end{array}
	\quad
	,~\chi_2(\theta_{-1}, o_1) =
	\left[
	\begin{array}{c}
	\theta_2 \setminus o_1 \\
	\theta_3 \setminus o_1 \\
	\ldots\\
	\theta_n \setminus o_1
	\end{array}
	\right]. \vspace*{.5ex}
\end{equation*} 
Here, $dec(o_1)=\sum_{j=1}^r 2^{j-1} \mathbf{I}_{j \in o_1}$ is a decimal 
index of bundle $o_1$, where $\mathbf{I}_{j\in o_1}=1$ if $j\in o_1$ and 
$\mathbf{I}_{j\in o_1} = 0$ otherwise. Attribute map $\chi_1$ thus stacks the 
vector $\theta_{-1}$, which represents the valuations of all agents except agent~$1$, with zero vectors of the same dimension, where the position of $\theta_{-1}$ is determined by the index of bundle $o_1$.
The resulting attribute vector is simple but potentially
restrictive. It precludes two instances with different allocated
bundles from sharing attributes, which provides an obstacle to
generalization of the discriminant function across bundles.
Attribute map $\chi_2$ stacks vectors $\theta_i\setminus o_1$, which are 
obtained from $\theta_i$ by setting the entries for all bundles
that intersect with $o_1$ to $0$. This captures the fact that agent $i$ cannot 
be allocated any of the bundles that intersect with $o_1$ if $o_1$ is allocated to agent~$1$.\footnote{Both
  $\chi_1$ and $\chi_2$ are defined for a particular number of items
  and agents, and in our experiments we train a different classifier
  for each number of agents and items. In practice, one can pad out
  items and agents by setting bids to zero and train a single
  classifier.}

\subsubsection{Efficient Computation of Inner Products}
Efficient computation of inner products is possible for both $\chi_1,
\chi_2$.  A full discussion can be found in \appref{app:inner-products}.

\subsubsection{Dealing with an Exponentially Large Output Space}
\label{subsubsec:exp-large-y}
Recall that Training Problems~1 and~2 have constraints for every training
example $(\theta^k, o_1^k)$ and every possible bundle of items $o_1 \in 
\Omega_1$, of which there are exponentially many in the number of items in 
the case of CAs. In lieu of an efficient separation oracle, a workaround 
exists when the discriminant function has additional structure, such 
that the induced payment weakly increases as items are added to a 
bundle. Given this \emph{item monotonicity}, it would suffice to include 
constraints for bundles that have a strictly larger value to the agent than 
any of their respective subsets.
Still, it remains an open problem whether item monotonicity itself can
be imposed on the hypothesis class with a small number of
constraints.\footnote{For polynomial kernels and certain attribute
  maps, a possible sufficient condition for item monotonicity is to
  force the weights $w_{-1}$ to be negative.  However, as with the
  discussion of enforcing $w_1 > 0$ directly, these weight constraints
  do not dualize conveniently and results in the dual formulation no
  longer operate on inner products $\langle\psi'(\theta_{-1}, o_1), \psi'(\theta_{-1}',
  o_1')\rangle$.  As a result, we would be forced to work in the primal, and
  incur extra computational overhead that increases polynomially with
  the kernel degree $d$.  We have performed some preliminary
  experiments with polynomial kernels, but we have not looked into
  reformulating the primal to enforce item monotonicity.}
An alternative is to optimistically assume item monotonicity, only
including the constraints associated with bundles that are explicit in
agent valuations.
The baseline experimental results in \secref{sec:experiments} do
not assume item monotonicity and instead use a separation oracle that
iterates over all possible bundles $o_1 \in \Omega_1$.  We also present
results which test the idea of optimistically assuming item
monotonicity, and while there is a degradation in performance, results
are mostly comparable.

\subsection{The Assignment Problem}
\label{subsec:assignment}

In the assignment problem, we are given a set of $n$ agents and
a set $\{1,\ldots,n\}$ of items, and wish to assign each item to exactly one 
agent. The outcome space of agent $i$ is thus $\Omega_i=\{1,\ldots,n\}$, and its type can be represented by a vector $\theta_i\in\Theta_i=\mathbb{R}^n$. The set of possible type profiles is then $\Theta=\mathbb{R}^{n^2}$.
We consider an outcome rule that maximizes {\em egalitarian welfare} in a lexicographic manner: first, the minimum value of any agent is maximized; if more than one outcome achieves the minimum, the second lowest value is maximized, and so forth.
This outcome rule can be computed by solving a sequence of integer programs.
As before, we assume agent symmetry and adopt the view of agent~$1$.

To complete our specification of the structural SVM framework for this
problem, we need to define an attribute map $\chi_3:\mathbb{R}^{n^2 - n} 
\times\mathbb{N}\rightarrow\mathbb{R}^s$, where the first argument is the type 
profile of all agents but agent~$1$, the second argument is the item assigned 
to agent~$1$, and~$s$ is a dimension of our choosing.
A natural choice for $\chi_3$ is to set
\begin{align*}
\chi_3(\theta_{-1}, j) &= (\theta_2[-j], \theta_3[-j], \ldots, \theta_n[-j]) 
\in \reals^{(n-1)^2},
\end{align*}
where $\theta_i[-j]$ denotes the vector obtained from $\theta_i$ by removing 
the $j$th entry. The attribute map thus reflects the agents' values for all 
items except item $j$, capturing the fact that the item assigned to agent~$1$ 
cannot be assigned to any other agent. Since the outcome space is very small, 
we choose not to use a non-linear kernel on top of this attribute vector.

\section{Experimental Evaluation}
\label{sec:experiments}

We perform a series of experiments to test our theoretical framework.
To run our experiments, we use the \svmstruct package~\citep{JFY09a},
which allows for the use of custom kernel functions, attribute maps,
and separation oracles.

\subsection{Setup}

We begin by briefly discussing our experimental methodology, performance metrics, and optimizations used to speed up the experiments.

\subsubsection{Methodology}

For each of the settings we consider, we generate three data sets: a training
set, a validation set, and a test set. The training set is used as input to
Training Problem~2, which in turn yields classifiers $h_w$ and corresponding
payment rules $p_w$.  For each choice of the parameter $C$ of Training
Problem~2, and the parameter $\gamma$ if the RBF kernel is used, a classifier
$h_w$ is learned based on the training set and evaluated based on the
validation set. The classifier with the highest accuracy on the validation set
is then chosen and evaluated on the test set.
During training, we take the perspective of agent~$1$, so a training set size 
of $\ell$ means that we train an SVM on $\ell$ examples.
Once a partial outcome rule has been learned, however, it can be used to 
infer payments for all agents. We exploit this fact during testing, and report performance metrics across all agents for a given instance in the test set.

\subsubsection{Metrics}

We employ three metrics to measure the performance of the learned
classifiers.  These metrics are computed over the test set
$\{(\theta^k, o^k)\}_{k=1}^\ell$.

\paragraph{Classification accuracy}
Classification accuracy measures the accuracy of the trained
classifier in predicting the outcome.  Each instance of the $\ell$
instances has $n$ agents, so in total we measure accuracy over $n\ell$
instances:\footnote{For a given instance $\theta$, there are actually
  many ways to choose $(\theta_i, \theta_{-i})$ depending on the
  ordering of all agents but agent~$i$.  We discuss a technique
  we refer to as sorting in \secref{subsubsec:exp-optimization},
  which will choose a particular ordering.  When this technique is not
  used, for example in our experiments for the assignment problem, we
  simply fix an ordering of the other agents for each agent $i$ and
  use the same ordering across all instances.} 
\[
	\text{{\em accuracy}} = 100 \cdot \frac{\sum_{k=1}^\ell \sum_{i=1}^n I(h_w(\theta_i, \theta_{-i}) = o_i^k))}{n \ell} .
\]
\paragraph{Ex post regret}
We measure ex post regret by summing over the ex post regret experienced by all
agents in each of the $\ell$ instances in the dataset, \ie
\[
\text{{\em regret}} = \frac{\sum_{k=1}^\ell \sum_{i=1}^n \regret_i(\theta_i^k,
\theta_{-i}^k)}{n\ell} .
\]
\paragraph{Individual rationality violation}
This metric measures the fraction of individual rationality violation
across all agents:
\begin{align*}
\text{{\em ir-violation}} = \frac{\sum_{k=1}^\ell \sum_{i=1}^n I(\irv_i(\theta_i, \theta_{-i}) >
0)}{n \ell}.
\end{align*}

\subsubsection{Optimizations}
\label{subsubsec:exp-optimization}

In the case of multi-minded CAs we map the inputs $\theta_{-1}$ into a smaller
space, which allows us to learn more effectively with smaller amounts of
data.\footnote{The barrier to using more data is not the availability of the
data itself, but the time required for training, because training time scales
quadratically in the size of the training set due to the use of non-linear
kernels.} We use {\em instance-based normalization}, which normalizes the values
in $\theta_{-1}$ by the highest observed value and then rescales the computed
payment appropriately, and {\em sorting}, which orders agents based on bid
values.

\paragraph{Instance-Based Normalization}
The first technique we use is \emph{instance-based
  normalization}. Before passing examples $\theta$ to the learning
algorithm or learned classifier, they are normalized by a positive
multiplier so that the value of the highest bid by agents other than
agent~$1$ is exactly $1$, before passing it to the learning algorithm
or classifier. The values and the solution are then transformed back
to the original scale before computing the payment rule $p_w$.  This
technique leverages the observation that agent 1's allocation depends
on the relative values of the other agent's reports (scaling all
reports by a factor should not affect the outcome chosen).

\paragraph{Sorting}
The second technique we use is \emph{sorting}.  With sorting, instead
of choosing an arbitrary ordering of agents in $\theta_{-i}$, we
choose a specific ordering based on the maximum value the agent
reports.  In the single-item setting, this amounts to ordering agents
by their value.  In the multi-minded CA setting, agents are ordered by
the value they report for their most desired bundle.  The intuition
behind sorting is that we can again decrease the space of possible
$\theta_{-i}$ reports the learner sees and learn more quickly.  In the
single-item case, we know that the second price payment rule only
depends on the maximum value across all other agents, and sorting
places this value in the first coordinate of $\theta_{-i}$.

\subsection{Single-Item Auction}

As a sanity check, we perform experiments on the single-item auction with the
optimal outcome rule, where the agent with the highest bid receives the item.
In the single-item case, we run experiments where $D$ is the distribution where
agent values are drawn independently and uniformly from $[0, 1]$.  The outcome
rule $g$ we use is the value-maximizing rule, i.e., the agent with the highest
value receives the item.  We use a training set size of 300 and validation and
test set sizes of 1000. In this case, we know that the associated payment
function that makes $(g, p)$ strategyproof is the second price payment rule.

The results reported in \tabref{tab:singleitem} and
\figref{fig:si_2_agents} are for the $\chi_1, \chi_2$ attribute
maps, which can be applied to this setting by observing that
single-item auctions are a special case of multi-minded CAs.
In particular, letting $z$
be the $0$ vector of dimension $n-1$, $\chi_1(\theta_{-1}, o_1) = (\theta_{-1},
z)$ if $o_1 = \emptyset$ and $\chi_1(\theta_{-1}, o_1) = (z, \theta_{-1})$ if 
$o_1= \{1 \}$ and $\chi_2(\theta_{-1}, o_1) = \theta_{-1}$ if $o_1=\emptyset$ 
and $\chi_2(\theta_{-1}, o_1)=z$ if $o_1 = \{1\}$.

For both choices of the attribute map we obtain excellent accuracy and very
close approximation to the second-price payment rule. This shows that the 
framework is able to automatically learn the payment rule of Vickrey's auction.

\begin{table}[ht!]
  \centering
  \begin{tabular}{*{7}{c}}
\hline \multirow{2}{*}{$n$} & \multicolumn{2}{c}{accuracy} & \multicolumn{2}{c}{regret} & \multicolumn{2}{c}{ir-violation}\\
 & $\chi_1$ & $\chi_2$ & $\chi_1$ & $\chi_2$ & $\chi_1$ & $\chi_2$\\
\hline
2& 99.7& 93.1& 0.000& 0.003& 0.00& 0.07\\
3& 98.7& 97.6& 0.000& 0.000& 0.01& 0.00\\
4& 98.4& 99.1& 0.000& 0.000& 0.00& 0.01\\
5& 97.3& 96.6& 0.001& 0.001& 0.02& 0.00\\
6& 97.6& 97.4& 0.000& 0.001& 0.00& 0.02\\\hline
\end{tabular}

  \caption{Performance metrics for single-item auction.\label{tab:singleitem}}
\end{table}
\begin{figure}[ht!]
  \centering
	\includegraphics[width=0.45\textwidth]{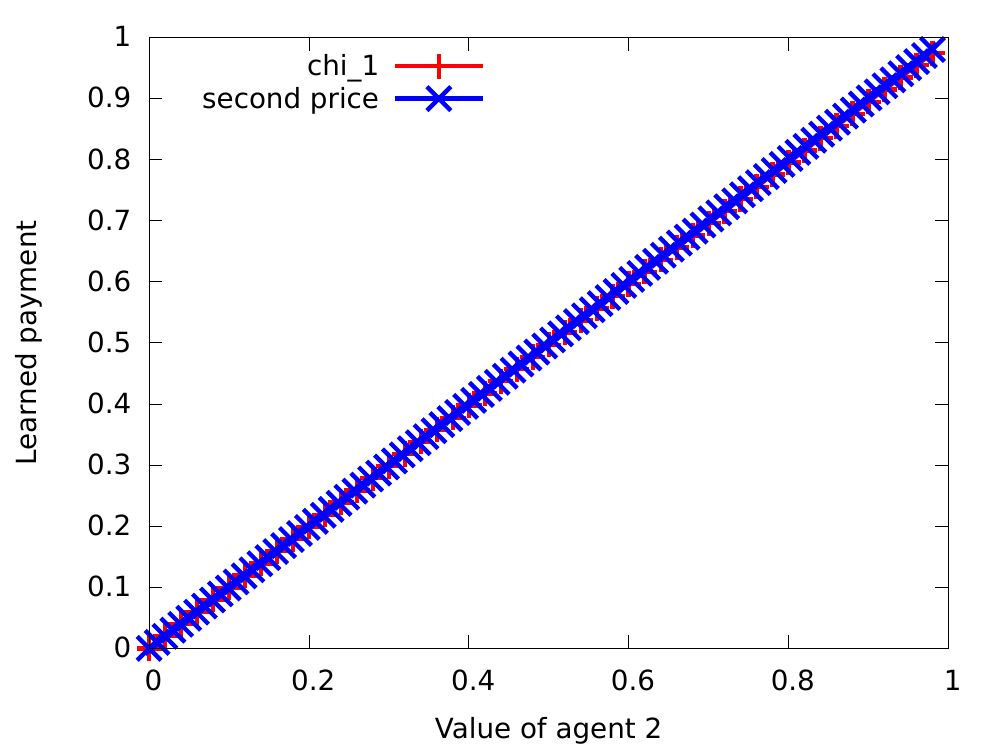} \hfill
	\includegraphics[width=0.45 \textwidth]{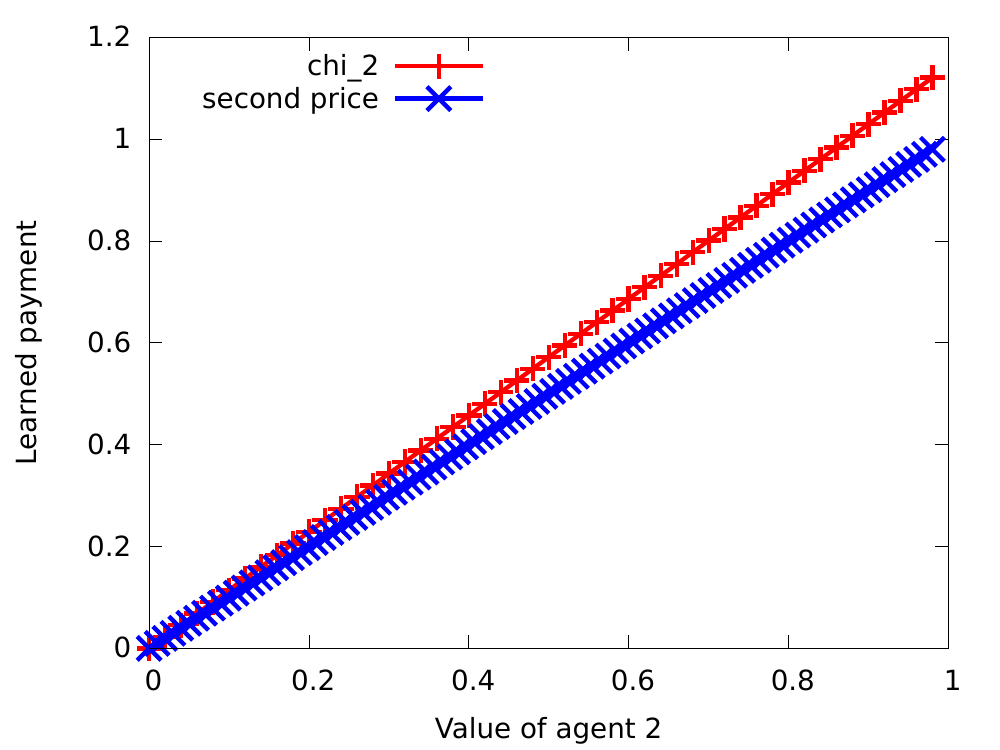}
  \caption{Learned payment rule vs.\@ second-price payment rule for single-item auction
  with $2$ agents, for $\chi_1$ (left) and $\chi_2$ (right).}
  \label{fig:si_2_agents}
\end{figure}

\subsection{Multi-Minded CAs}
\subsubsection{Type Distribution}

Recall that in a multi-minded setting, there are $r$ items, and each
agent is interested in exactly $b$ bundles.  For each bundle, we use
the following procedure (inspired by Sandholm's decay distribution for
the single-minded setting \cite{Sandholm02}) to determine which items
are included in the bundle.  We first assign an item to the bundle
uniformly at random.  Then with probability $\alpha$, we add another
random item (chosen uniformly from the remaining items), and with
probability $(1 - \alpha)$ we stop.  We continue this procedure until
we stop or have exhausted the items.  We use $\alpha = 0.75$ to be
consistent with \cite{Sandholm02}, as they report that the winner
determination problem (finding the feasible allocation that maximizes
total value) is difficult for this setting of $\alpha$.

Once the bundle identities have been determined, we sample values for
these bundles.  Let $\bc$ be an $r$-dimensional vector with entries
chosen uniformly from $(0,1]$. For each agent~$i$, let $\bd_i$ be an
  $r$-dimensional vector with entries chosen uniformly from
  $(0,1]$. Each entry of $\bc$ denotes the common value of a specific
    item, while each entry of $\bd_i$ denotes the private value of a
    specific item for agent~$i$. The value of bundle $S_{ij}$ is then
    given by
\begin{equation*}
	v_{ij} = \min_{S_{ij'}\leq S_{ij}} \left( \frac{\langle S_{ij'},
  \beta\bc+(1-\beta)\bd_i \rangle}{r} \right)^\zeta
\end{equation*}
for parameters $\beta\in[0,1]$ and $\zeta$. The inner product in the
numerator corresponds to a sum over values of items, where common and
private values for each item are respectively weighted with $\beta$
and $(1-\beta)$. The denominator normalizes all valuations to the
interval $(0,1]$. Parameter $\zeta$ controls the degree of
  complementarity among items: $\zeta>1$ implies that goods are
  complements, whereas $\zeta<1$ means that goods are
  substitutes. Choosing the minimum over bundles $S_{ij'}$ contained
  in $S_{ij}$ finally ensures that the resulting valuations are
  monotonic.

\subsubsection{Outcome Rules}

We use two outcome rules in our experiments.  For the \textit{optimal}
outcome rule, the payment rule $p_{vcg}$ makes the mechanism
$(g_{opt}, p_{vcg})$ strategyproof.  Under this payment rule, agent
$i$ pays the externality it imposes on other agents.  That is, 
$$p_{vcg,1}(\theta) = \left(\max_{o \in \Omega} \sum_{i \ne 1} v_i(\theta_i, o_i) \right) - \sum_{i \ne 1}
v_i(\theta_i, g_i(\theta)).$$

The second outcome rule with which we experiment is a
generalization of the greedy outcome rule for single-minded CA
\citet{Lehmann2002}.  Our generalization of the greedy rule is as
follows.  Let $\theta$ be the agent valuations and $o_i(j)$ denote the
$j$-th bundle desired by agent $i$.  For each bundle $o_i(j)$, assign
a score $v_i(\theta_i, o_i(j)) / \sqrt{|o_i(j)|}$, where $|o_i(j)|$ indicates the
total items in bundle $o_i(j)$.  The greedy outcome rule orders the
desired bundles by this score, and takes the bundle $o_i(j)$ with the
next highest score as long as agent $i$ has not already been allocated
a bundle and $o_i(j)$ does not contain any items already allocated.
While this greedy outcome rule has an associated payment rule that
makes it strategyproof in the single-minded case, it is not implementable
in the multi-minded case as the example in \appref{app:greedy}
shows.

\subsubsection{Description of Experiments}
We experiment with training sets of sizes $100$, $300$, and $500$, and
validation and test sets of size $1000$. All experiments we report on are for 
a setting with $5$ agents, $5$ items, and $3$
bundles per agent, and use $\beta=0.5$, the RBF kernel, and parameters $C\in\{10^4,10^5\}$ and $\gamma\in\{0.01,0.1,1\}$.

\subsubsection{Basic Results}
\begin{table}
  \centering
  {\small
\begin{tabular}{@{}c@{\hspace{3.3pt}}c@{\hspace{2pt}}
  *{9}{@{\hspace{3.3pt}}c}@{\hspace{2pt}}
  *{9}{@{\hspace{3.3pt}}c}@{}}
\hline 
 &  &
\multicolumn{9}{c}{Optimal outcome rule} & 
\multicolumn{9}{c}{Greedy outcome rule} \\
 &  &
\multicolumn{3}{c}{accuracy} & 
\multicolumn{3}{c}{regret} & 
\multicolumn{3}{c}{ir-violation} &
\multicolumn{3}{c}{accuracy} & 
\multicolumn{3}{c}{regret} & 
\multicolumn{3}{c}{ir-violation} \\
      $n$ & $\zeta$ & 
      $p_{vcg}$ & $\chi_1$ & $\chi_2$ & 
      $p_{vcg}$ & $\chi_1$ & $\chi_2$ & 
      $p_{vcg}$ & $\chi_1$ & $\chi_2$ &
      $p_{vcg}$ & $\chi_1$ & $\chi_2$ & 
      $p_{vcg}$ & $\chi_1$ & $\chi_2$ & 
      $p_{vcg}$ & $\chi_1$ & $\chi_2$ \\
				\hline
        2 & 0.5& 100& 70.7& 91.9& 0& 0.014& 0.002& 0.0& 0.06& 0.03 &  50.9& 59.1& 40.6& 0.079& 0.030& 0.172& 0.22& 0.12& 0.33\\
3 & 0.5& 100& 54.5& 75.4& 0& 0.037& 0.017& 0.0& 0.19& 0.10 &  55.4& 57.9& 54.7& 0.070& 0.030& 0.088& 0.18& 0.21& 0.36\\
4 & 0.5& 100& 53.8& 67.7& 0& 0.042& 0.031& 0.0& 0.22& 0.18 &  61.1& 58.2& 57.9& 0.056& 0.033& 0.056& 0.14& 0.20& 0.31\\
5 & 0.5& 100& 15.8& 67.0& 0& 0.133& 0.032& 0.0& 0.26& 0.19 &  64.9& 61.3& 63.0& 0.048& 0.027& 0.042& 0.13& 0.19& 0.24\\
6 & 0.5& 100& 61.1& 68.2& 0& 0.037& 0.032& 0.0& 0.22& 0.20 &  66.6& 63.8& 63.8& 0.041& 0.034& 0.045& 0.12& 0.20& 0.24\\\hline
2 & 1.0& 100& 84.5& 93.4& 0& 0.008& 0.001& 0.0& 0.08& 0.02 &  87.8& 86.6& 84.0& 0.007& 0.005& 0.008& 0.04& 0.06& 0.09\\
3 & 1.0& 100& 77.1& 83.5& 0& 0.012& 0.005& 0.0& 0.13& 0.09 &  85.3& 86.7& 85.7& 0.006& 0.006& 0.006& 0.04& 0.07& 0.05\\
4 & 1.0& 100& 74.6& 81.1& 0& 0.014& 0.009& 0.0& 0.16& 0.12 &  82.4& 86.5& 84.2& 0.006& 0.006& 0.007& 0.05& 0.08& 0.08\\
5 & 1.0& 100& 73.4& 77.4& 0& 0.018& 0.011& 0.0& 0.19& 0.12 &  82.7& 85.8& 84.9& 0.007& 0.009& 0.009& 0.04& 0.10& 0.10\\
6 & 1.0& 100& 75.0& 77.7& 0& 0.020& 0.013& 0.0& 0.20& 0.16 &  80.0& 87.4& 88.1& 0.006& 0.007& 0.005& 0.04& 0.08& 0.07\\\hline
2 & 1.5& 100& 91.5& 96.9& 0& 0.004& 0.000& 0.0& 0.06& 0.02 &  94.7& 91.1& 91.7& 0.002& 0.002& 0.002& 0.02& 0.04& 0.04\\
3 & 1.5& 100& 91.0& 93.4& 0& 0.004& 0.001& 0.0& 0.05& 0.03 &  97.1& 92.8& 93.2& 0.001& 0.002& 0.001& 0.01& 0.02& 0.04\\
4 & 1.5& 100& 92.5& 94.2& 0& 0.003& 0.001& 0.0& 0.03& 0.04 &  96.4& 91.5& 92.1& 0.001& 0.003& 0.002& 0.02& 0.07& 0.07\\
5 & 1.5& 100& 91.7& 93.9& 0& 0.004& 0.002& 0.0& 0.06& 0.03 &  97.5& 90.5& 91.4& 0.001& 0.004& 0.002& 0.01& 0.06& 0.04\\
6 & 1.5& 100& 91.9& 93.7& 0& 0.003& 0.001& 0.0& 0.05& 0.04 &  98.4& 92.2& 92.8& 0.000& 0.003& 0.002& 0.01& 0.06& 0.06\\\hline

\end{tabular}

  \caption{Results for multi-minded CA with
           training set size 500.\label{tbl:basic-combined}}
  }
\end{table}
\tabref{tbl:basic-combined} presents the basic
results for multi-minded CAs with optimal and greedy outcome rules,
respectively. For both outcome rules, we present the results for $p_{vcg}$ as a
baseline.
Because $p_{vcg}$ is the
strategyproof payment rule for the optimal outcome rule, $p_{vcg}$ always has
accuracy~$100$, regret~$0$, and IR violation~$0$ for the optimal outcome rule.

Across all instances, as expected, accuracy is negatively correlated with regret
and ex post IR violation.  The degree of complementarity between items, $\zeta$,
as well as the outcome rule chosen, has a major effect on the results.
Instances with low complementarity ($\zeta = 0.5$) yield payment rules with
higher regret, and $\chi_1$ performs better on the greedy outcome rule while
$\chi_2$ performs better on the optimal outcome rule.  For high complementarity
between items the greedy outcome tends to allocate all items to a single agent,
and the learned price function sets high prices for small bundles to capture
this property. For low complementarity the allocation tends to be split and less
predictable.
Still, the best classifiers achieve average ex post regret of less than 0.032
(for values normalized to [0,1]) even though the corresponding prediction
accuracy can be as low as 67\%.  For the greedy outcome rule, the performance of
$p_{vcg}$ is comparable for $\zeta\in\{1.0,1.5\}$ but worse than the payment
rule learned in our framework in the case of $\zeta=0.5$, where the greedy
outcome rule becomes less optimal.
\begin{table}[htbp]
\centering
\footnotesize{
\begin{tabular}{@{}*{16}{@{\hspace{7pt}}c}@{}}
\hline \multirow{2}{*}{$n$} & \multirow{2}{*}{$\zeta$}&accuracy&\multicolumn{2}{c}{100}&\multicolumn{2}{c}{300}&\multicolumn{2}{c}{500}&
regret&\multicolumn{2}{c}{100}&\multicolumn{2}{c}{300}&\multicolumn{2}{c}{500}\\
& & $p_{vcg}$ & $\chi_1$  & $\chi_2$ & $\chi_1$ & $\chi_2$ & $\chi_1$ & $\chi_2$ & $p_{vcg}$ & $\chi_1$ & $\chi_2$ & $\chi_1$ & $\chi_2$ & $\chi_1$ & $\chi_2$\\
\hline
2 & 0.5& 50.9& 54.3& 48.2& 57.0& 46.9& 59.1& 40.6& 0.079& 0.045& 0.195& 0.032& 0.098& 0.030& 0.172\\
3 & 0.5& 55.4& 50.1& 49.8& 55.7& 54.4& 57.9& 54.7& 0.070& 0.054& 0.078& 0.038& 0.082& 0.030& 0.088\\
4 & 0.5& 61.1& 53.4& 56.2& 56.4& 58.5& 58.2& 57.9& 0.056& 0.050& 0.059& 0.040& 0.061& 0.033& 0.056\\
5 & 0.5& 64.9& 14.2& 57.9& 61.0& 61.8& 61.3& 63.0& 0.048& 0.173& 0.064& 0.038& 0.048& 0.027& 0.042\\
6 & 0.5& 66.6& 58.4& 58.8& 62.2& 63.9& 63.8& 63.8& 0.041& 0.039& 0.059& 0.037& 0.049& 0.034& 0.045\\\hline
2 & 1.0& 87.8& 80.7& 80.5& 84.4& 84.1& 86.6& 84.0& 0.007& 0.010& 0.010& 0.009& 0.008& 0.005& 0.008\\
3 & 1.0& 85.3& 74.9& 78.0& 83.0& 80.6& 86.7& 85.7& 0.006& 0.020& 0.011& 0.009& 0.009& 0.006& 0.006\\
4 & 1.0& 82.4& 78.5& 80.1& 84.2& 83.1& 86.5& 84.2& 0.006& 0.015& 0.014& 0.008& 0.009& 0.006& 0.007\\
5 & 1.0& 82.7& 81.0& 81.8& 84.3& 84.3& 85.8& 84.9& 0.007& 0.020& 0.014& 0.010& 0.009& 0.009& 0.009\\
6 & 1.0& 80.0& 81.8& 83.7& 87.6& 88.3& 87.4& 88.1& 0.006& 0.062& 0.018& 0.008& 0.005& 0.007& 0.005\\\hline
2 & 1.5& 94.7& 83.3& 88.1& 89.3& 89.8& 91.1& 91.7& 0.002& 0.008& 0.003& 0.003& 0.002& 0.002& 0.002\\
3 & 1.5& 97.1& 86.9& 87.6& 90.3& 91.5& 92.8& 93.2& 0.001& 0.005& 0.004& 0.003& 0.002& 0.002& 0.001\\
4 & 1.5& 96.4& 88.4& 90.7& 89.3& 90.8& 91.5& 92.1& 0.001& 0.005& 0.003& 0.004& 0.003& 0.003& 0.002\\
5 & 1.5& 97.5& 87.2& 88.5& 91.4& 90.5& 90.5& 91.4& 0.001& 0.006& 0.004& 0.003& 0.003& 0.004& 0.002\\
6 & 1.5& 98.4& 86.3& 86.8& 91.4& 92.5& 92.2& 92.8& 0.000& 0.011& 0.007& 0.004& 0.002& 0.003& 0.002\\\hline

\end{tabular}
}
\caption{Effect of training set size on accuracy of learned classifier.
Multi-minded CA, greedy outcome rule.  Training set size is given in the column labels for $\chi_1, \chi_2$.  $p_{vcg}$
does not have a training set size.
\label{tbl:training-set-size}}
\end{table}
\subsubsection{Effect of Training Set Size}
Table \ref{tbl:training-set-size} charts
performance as the training set size is varied for the greedy outcome rule.
While training data is readily available (we can simply sample from $D$ and run
the outcome rule $g$), training time becomes prohibitive for larger training set
sizes.  Table \ref{tbl:training-set-size} shows that regret decreases with
larger training sets, and for a training set size of 500, the best of $\chi_1$
and $\chi_2$ outperforms $p_{vcg}$ for $\zeta = 0.5$ and is comparable to
$p_{vcg}$ for $\zeta \in \{ 1.0, 1.5 \}$.

\subsubsection{IR Fixes}
\begin{table}
\centering
\newcommand{\thead}[1]{\multicolumn{3}{c}{\makebox[0pt]{#1}}}
\begin{tabular*}{.9\linewidth}{*{13}{@{\extracolsep{\fill}}c}@{}}
  \hline
	\multirow{2}{*}{\parbox{1.1cm}{\centering payment offset}} 
	& \thead{accuracy} & \thead{regret} & \thead{ir-violation} & 
	\thead{ir-fix-welfare-avg} \\
 	& 0.5 & 1.0 & 1.5 & 0.5 & 1.0 & 1.5 & 0.5 & 1.0 & 1.5 & 0.5 & 1.0 & 1.5 \\
	\hline
  0& 59.7& 61.8& 61.7& 0.065& 0.048& 0.042& 0.35& 0.26& 0.21& 0.27& 0.43& 0.52\\
0.05& 61.7& 61.2& 60.1& 0.054& 0.045& 0.044& 0.29& 0.20& 0.15& 0.37& 0.54& 0.65\\
0.10& 62.1& 59.3& 56.7& 0.048& 0.047& 0.051& 0.23& 0.14& 0.10& 0.48& 0.66& 0.75\\
0.15& 60.4& 55.1& 52.2& 0.047& 0.055& 0.064& 0.17& 0.10& 0.06& 0.59& 0.75& 0.84\\
0.20& 57.8& 51.7& 48.5& 0.052& 0.067& 0.079& 0.12& 0.06& 0.03& 0.70& 0.83& 0.90\\
0.25& 54.3& 47.7& 44.3& 0.061& 0.082& 0.096& 0.08& 0.03& 0.02& 0.79& 0.89& 0.93\\\hline

\end{tabular*}

\caption{Impact of payment offset and null loss fix for $\zeta=0.5$ and greedy 
	       outcome rule, training set size 300. All results are for $\chi_2$,
	       null loss values appear in the second row. \label{tbl:dealloc-fix}}
\end{table}
\begin{figure}
	\centering
  \includegraphics[width=0.45\textwidth,page=4]{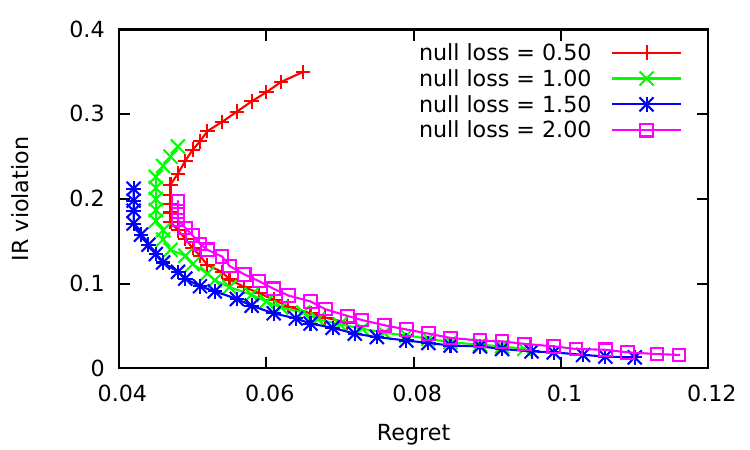}
  \includegraphics[width=0.45\textwidth,page=9]{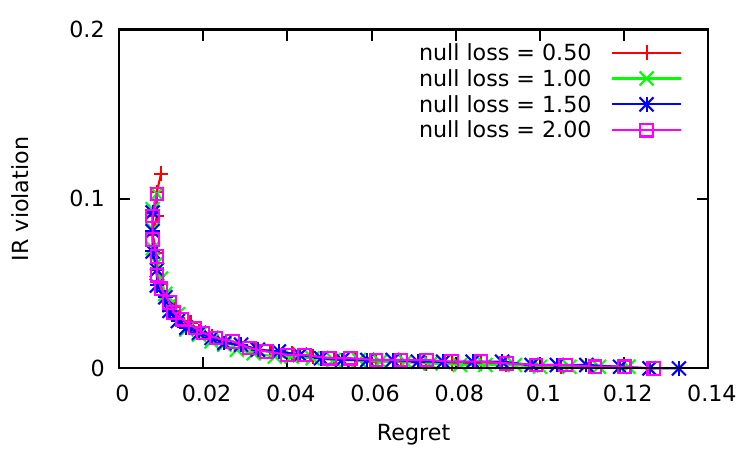}
	\vspace*{-1ex}
	\caption{Impact of payment offset and null loss fix for
		greedy outcome rule, training set size~300.\label{fig:ir-fixes}}
\end{figure}
\tabref{tbl:dealloc-fix} summarizes our results regarding the various fixes to 
IR violations, for the particularly challenging case of the greedy outcome 
rule and $\zeta = 0.5$.
The extent of IR violation decreases with larger payment offset and null 
loss.  Regret tends to move in the opposite direction, but there are cases 
where IR violation and regret both decrease.
The three rightmost columns of \tabref{tbl:dealloc-fix} list the average 
ratio between welfare after and before the deallocation fix, across the
instances in the test set.  With a payment offset of~$0$, a large
welfare hit is incurred if we deallocate agents with IR violations.
However, this penalty decreases with increasing payment offsets and increasing 
null loss. At the most extreme payment offset and null loss adjustment, the IR 
violation is as low as $2\%$, and the deallocation fix incurs a welfare 
loss of only~$7\%$.

\figref{fig:ir-fixes} shows a graphical representation of the
impact of payment offsets and null losses. Each line in the plot corresponds 
to a payment rule learned with a different null loss, and each point on a
line corresponds to a different payment offset. The payment offset is zero for 
the top-most point on each line, and equal to~$0.29$ for the lowest point on 
each line. Increasing the payment offset always decreases the rate of IR 
violation, but may decrease or increase regret. Increasing null loss lowers 
the top-most point on a given line, but arbitrarily increasing null loss can 
be harmful. Indeed, in the figure on the left,
a null loss of~$1.5$ results in a slightly higher 
top-most point but significantly lower regret at this top-most point compared 
to a null loss of~$2.0$.  It is also interesting to note that these adjustments
have much more impact on the hardest distribution with $\zeta = 0.5$.
\begin{table}
\centering
\begin{tabular}{@{}
  *{8}{@{\hspace{4pt}}c}@{}}
\hline \multirow{2}{*}{$n$} & \multirow{2}{*}{$\zeta$} & \multicolumn{2}{c}{accuracy} & \multicolumn{2}{c}{regret} & \multicolumn{2}{c}{ir-violation}\\
 &  & $\chi_2$ & $\chi_2$ (i-mon) & $\chi_2$ & $\chi_2$ (i-mon) & $\chi_2$ & $\chi_2$ (i-mon)\\
\hline
2 & 0.5& 46.9& 46.3& 0.098& 0.232& 0.28& 0.38\\
3 & 0.5& 54.4& 8.6& 0.082& 0.465& 0.33& 0.06\\
4 & 0.5& 58.5& 48.2& 0.061& 0.811& 0.31& 0.25\\
5 & 0.5& 61.8& 57.0& 0.048& 0.136& 0.26& 0.26\\
6 & 0.5& 63.9& 61.3& 0.049& 0.078& 0.25& 0.20\\\hline
2 & 1.0& 84.1& 82.2& 0.008& 0.010& 0.06& 0.08\\
3 & 1.0& 80.6& 80.1& 0.009& 0.010& 0.10& 0.09\\
4 & 1.0& 83.1& 79.7& 0.009& 0.012& 0.11& 0.11\\
5 & 1.0& 84.3& 77.2& 0.009& 0.020& 0.10& 0.11\\
6 & 1.0& 88.3& 83.9& 0.005& 0.013& 0.08& 0.11\\\hline
2 & 1.5& 89.8& 89.1& 0.002& 0.003& 0.03& 0.06\\
3 & 1.5& 91.5& 91.3& 0.002& 0.003& 0.04& 0.04\\
4 & 1.5& 90.8& 89.7& 0.003& 0.003& 0.06& 0.06\\
5 & 1.5& 90.5& 87.3& 0.003& 0.005& 0.04& 0.05\\
6 & 1.5& 92.5& 70.8& 0.002& 0.081& 0.06& 0.17\\\hline
\end{tabular}

\caption{Comparison of performance with and without optimistically assuming
item monotonicity.  (i-mon) indicates a
payment rule learned by optimistically assuming item
monotonicity. Greedy outcome rule. Training set size 300. \label{tbl:item-mon}}
\end{table}
\subsubsection{Item Monotonicity}
\tabref{tbl:item-mon} presents a comparison of a payment rule
learned with explicit enumeration of all bundle constraints (the
default that we have been using for our other results) and a payment
rule learned by optimistically assuming item monotonicity (see 
\secref{subsubsec:exp-large-y}).  Performance is affected when we drop
constraints and optimistically assume item monotonicity, although the
effects are small for $\zeta \in \{ 1.0, 1.5 \}$ and larger for $\zeta
5= 0.5$.  Because item monotonicity allows for the training problem to
be succinctly specified, we may be able to train on more data, and
this seems a very promising avenue for further consideration (perhaps
coupled with heuristic methods to add additional constraints to the
training problem).

\subsection{The Assignment Problem}
\label{subsec:exp-assignment}

In the assignment problem, agents' values for the items are sampled uniformly and independently from $[0,1]$. We use a training set of size $600$, validation and test sets of size $1000$, and the RBF kernel with parameters
$C\in\{10,1000,100000\}$ and $\gamma\in\{0.1,0.5,1.0\}$.

The performance of the learned payment rules is compared to that of three 
VCG-based payment rules. Let $W$ be the total welfare of all agents other than 
$i$ under the outcome chosen by $g$, and $W_{eg}$ be the minimum value any 
agent other than $i$ receives under this outcome. We then consider the 
following payment rules:
(1)~the {\em vcg} payment rule, where agent~$i$ pays the difference
between the maximum total welfare of the other agents under any allocation
and $W$;
(2)~the {\em tot-vcg} payment rule, where agent~$i$ pays the
difference between the total welfare of the other agents under the allocation 
maximizing egalitarian welfare and $W$; and
(3)~the {\em eg-vcg} payment rule, where agent~$i$ pays the
difference between the minimum value of any agent under the allocation
maximizing egalitarian welfare and $W_{eg}$.
\begin{table}
\centering
\newcommand{\thead}[1]{\multicolumn{4}{c}{\makebox[0pt]{#1}}}
\begin{tabular*}{.85\linewidth}{*{13}{@{\extracolsep{\fill}}c}@{}}
  \hline 
	\multirow{2}{*}{$n$} & 
	\thead{accuracy} & \thead{regret} & \thead{ir-violation} \\
  & vcg & tot-vcg & eg-vcg & $p_w$ & vcg & tot-vcg & eg-vcg & $p_w$ & vcg &
  tot-vcg & eg-vcg & $p_w$ \\
  \hline
  2& 64.3& 67.5& 67.5& 89.0& 0.018& 0.015& 0.015& 0.023& 0.03& 0.01& 0.01& 0.03\\
3& 48.0& 52.1& 42.5& 77.9& 0.070& 0.077& 0.127& 0.041& 0.06& 0.07& 0.03& 0.04\\
4& 40.6& 43.1& 30.8& 71.0& 0.111& 0.123& 0.199& 0.054& 0.07& 0.09& 0.03& 0.02\\
5& 32.4& 35.3& 24.5& 63.9& 0.157& 0.169& 0.254& 0.071& 0.10& 0.12& 0.03& 0.01\\
6& 27.1& 29.9& 20.0& 59.0& 0.189& 0.208& 0.290& 0.074& 0.10& 0.13& 0.03& 0.01\\\hline

\end{tabular*}

\caption{Results for assignment problem with egalitarian outcome rule 
\label{tbl:eg-welfare}}
\end{table}

The results for attribute map $\chi_3$ are shown in \tabref{tbl:eg-welfare}.  
We see that the learned payment rule $p_w$ yields significantly lower regret 
than any of the VCG-based payment rules, and average ex post regret less than 
$0.074$ for values normalized to $[0,1]$. Since we are not maximizing the sum 
of values of the agents, it is not very surprising that VCG-based payment 
rules perform rather poorly. The learned payment rule $p_w$ can adjust to the 
outcome rule, and also achieves a low fraction of ex post IR
violation of at most $3\%$.
\jl{Fix this to reflect the actual data.}

\section{Conclusions}

We have introduced a new paradigm for computational mechanism design
in which statistical machine learning is adopted to design payment rules 
for given algorithmically specified outcome rules, and have
shown encouraging experimental results. Future directions of interest
include (1)~an alternative formulation of the problem as a regression
rather than classification problem, (2)~constraints on properties of the
learned payment rule, concerning for example the core or budgets, 
(3)~methods that learn classifiers more likely to induce feasible
outcome rules, so that these learned outcome rules can be used, 
(4)~optimistically
assuming item monotonicity and dropping constraints implied by it, thereby allowing for better scaling of training time with
training set size at the expense of optimizing against a subset of
the full constraints in the training problem,
and (5)~an investigation of the extent to which alternative goals such as
regret percentiles or {\em interim} regret can be achieved through
machine learning.

\section*{Acknowledgments} 

We thank Shivani Agarwal, Vince Conitzer, Amy Greenwald, Jason
Hartline, and Tim Roughgarden for valuable discussions and the anonymous 
referees for helpful feedback.  All errors remain our own. This material is based upon work supported in part by the National Science Foundation under grant CCF-1101570, the Deutsche Forschungsgemeinschaft under grant FI~1664/1-1, an EURYI award, and an NDSEG fellowship.
\bibliography{abbshort,svmmd}

\appendix

\section{Efficient Computation of Inner Products} %
\label{app:inner-products}
For both $\chi_1$ and $\chi_2$, computing inner products reduces to
the question of whether inner products between valuation profiles are
efficiently computable. For $\chi_1$, we have that
\begin{align*}
  \left\langle \chi_1(\theta_{-1}, o_1), \chi_1(\theta_{-1}', o_1') \right\rangle
  &= \mathbf{I}_{o_1 = o_1'} \sum_{i=2}^n \left\langle \theta_i, \theta'_i \right\rangle , 
	\intertext{where indicator $\mathbf{I}_{o_1=o_1'}=1$ if $o_1=o_1'$ and $\mathbf{I}_{o_1 = o_1'}=0$ otherwise. For $\chi_2$,}
	\left\langle \chi_2(\theta_{-1}, o_1), \chi_2(\theta_{-1}', o_1') \right\rangle &=
  \sum_{i=2}^n \left\langle \theta_i \setminus o_1, \theta'_i \setminus o_1 \right\rangle .
\end{align*}

We next develop efficient methods for computing the inner products
$\left\langle \theta_i,\theta'_i\right\rangle$ on compactly represented
valuation functions.  The computation of $\left\langle \theta_i\setminus
o_1,\theta'_i\setminus o_1 \right\rangle$ can be done through similar methods.

In the single-minded setting, let $\theta_i$ correspond to a bundle
$S_i\subseteq \{1,\ldots,r\}$ of items with value $v_i$, and $\theta'_i$
correspond to a set $S'_i\subseteq \{1,\ldots,r\}$ of items valued at
$v'_i$.

Each set containing both $S_i$ and $S'_i$ contributes $v_i v'_i$ to
$\theta_i^T\theta'_i$, while all other sets contribute $0$. Since there are
exactly $2^{r-|S_i \cup S'_i|}$ sets containing both $S_i$ and $S'_i$,
we have
\[
	\theta_i^T \theta'_i = v_i v'_i 2^{r-|S_i \cup S'_i|}.
\]
 
This is a special case of the formula for the multi-minded case. 
\begin{lemma}  \label{lem:dot-product}
	Consider a multi-minded CA and two bid vectors $x_1$ and
        $x'_1$ corresponding to sets $S=\{S_1,\dots,S_s\}$ and
        $S'=\{S'_1,\dots,S'_t\}$, with associated values $v_1, \dots,
        v_s$ and $v'_1,\dots,v'_t$. Then,
\begin{multline} \label{eq:iprod}
	x_1^T x'_1 = \sum_{T\subseteq S, T'\subseteq S'} \Bigl(
        (-1)^{|T|+|T'|} \cdot (\min_{S_i \in T} v_i) \cdot (\min_{S'_j
          \in T'} v'_j) \cdot %
        2^{r-|(\bigcup_{S_i \in T} S_i) \cup (\bigcup_{S'_j \in T'}
          S'_j)|} \Bigr).
\end{multline}
\end{lemma}
\begin{proof}
	The contribution of a particular bundle $B$ of items to the
        inner product is $(\max_{S_i \in S, S_i \subseteq B} v_i)
        \cdot (\max_{S'_j \in S', S'_j \subseteq B} v'_j)$, and thus
	\[
		x_1^T x'_1 = \sum_{B} \Bigl((\max_{\nfrac{S_i\in S}{S_i \subseteq B}} v_i) \cdot (\max_{\nfrac{S'_j\in S'}{S'_j\subseteq B}} v'_j)\Bigr).
	\]
	By the maximum-minimums identity, 
        which asserts that for any set $\{x_1, \dots, x_n\}$ of $n$ numbers, $\max\{x_1,\dots,x_n\} = \sum_{Z \subseteq X} ((-1)^{|Z|+1} \cdot (\min_{x_i \in Z}x_i))$, 
	\begin{align*}
		&\max_{\nfrac{S_i \in S}{S_i \subseteq B}} v_i = \sum_{\nfrac{T\subseteq S}{\mathclap{\bigcup_{S_i\in T}S_i\subseteq B}}} \Bigl((-1)^{|T|+1} \cdot (\min_{S_i \in T} v_i) \Bigr) \quad\text{and} \\
		&\max_{\nfrac{S'_j \in S'}{S'_j \subseteq B}} v'_j = \sum_{\nfrac{T' \subseteq S'}{\mathclap{\bigcup_{S'_j\in T'}S'_j\subseteq B}}} \Bigl( (-1)^{|T'|+1} \cdot (\min_{S'_j \in T'} v'_j) \Bigr).
	\end{align*}
	The inner product can thus be written as  
	\[
		\theta_1^T \theta'_1 = \sum_{B} \hspace{-.75em} \sum_{\nfrac{T\subseteq S, T'\subseteq S'}{\nfrac{\bigcup_{S_i\in T} S_i \subseteq B}{\bigcup_{S'_j \in T'} S'_j \subseteq B}}} \hspace{-1em}\Bigl( (-1)^{|T|+|T'|} \cdot (\min_{S_i \in T} v_i) \cdot (\min_{S'_j \in T'} v'_j) \Bigr).
	\]
	Finally, for given $T\subseteq S$ and $T'\subseteq S'$, there
        exist exactly $2^{r-|(\bigcup_{S_i \in T} S_i) \cup
          (\bigcup_{S'_j \in T'} S'_j)|}$ bundles $B$ such that
        $\bigcup_{S_i \in T} S_i \subseteq B$ and $\bigcup_{S'_j \in
          T'} S'_j \subseteq B$, and we obtain
	\begin{multline*}
		\theta_1^T \theta'_1 = \sum_{T \subseteq S, T' \subseteq S'} \Bigl( (-1)^{|T|+|T'|} \cdot (\min_{S_i \in T} v_i) \cdot (\min_{S'_j \in T'} v'_j) \cdot %
		2^{m-|(\bigcup_{S_i \in T} S_i) \cup (\bigcup_{S'_j \in T'} S'_j)|} \Bigr).
	\end{multline*}
\end{proof}

If $S$ and $S'$ have constant size, then the sum on the right hand
side of~\eqref{eq:iprod} ranges over a constant number of sets and can
be computed efficiently.

\section{Greedy Allocation Rule is not Weakly Monotone}
\label{app:greedy}

Consider a setting with a single agent and four items.

If the valuations $\theta_1$ of the agent are
\begin{align*}
	&v_1(\theta_1, o_1) = \begin{cases}
                  20 &\text{if $o_1 = \{ 1, 2, 3, 4 \}$}\\
                  12 &\text{if $1 \in o_1$ and $j \notin o_1$ for some $j \in \{2,3,4\}$, and}\\
                  0  &\text{else}
                \end{cases}
\end{align*}
then the allocation is $ \{ 1 \}$.

If the valuations are $\theta_1'$ such that
\begin{align*}
	&v_1(\theta_1',o_1) =  \begin{cases}
                     12  &\text{if $o_1 = \{ 1, 2, 3, 4 \}$}\\
                      5  &\text{if $1 \in o_1$ and $j \notin o_1$ for some $j \in \{2,3,4\}$, and}\\
                      0  &\text{else}
                   \end{cases}
\end{align*}
then the allocation is $\{ 1, 2, 3, 4 \}$.

We have $v_1(\theta_1', \{1, 2, 3, 4\}) - v_1(\theta_1', \{1\}) < v_1(\theta_1, \{1, 2, 3, 4\}) - v_1(\theta_1, \{ 1\})$
contradicting weak monotonicity.

\end{document}